\documentclass[a4paper]{article}
\usepackage{amsmath,amsfonts,amsthm,amssymb}
\usepackage[usenames,dvipsnames]{color}
\usepackage{tikz-cd}
\usetikzlibrary{matrix,arrows,calc}
\usepackage{calc}
\usepackage[colorlinks=true,linkcolor=blue,citecolor=red,urlcolor=Violet,bookmarksdepth=subsection,final]{hyperref}

\newtheorem{thm}{Theorem}[section]
\newtheorem{prop}[thm]{Proposition}
\newtheorem{lem}[thm]{Lemma}

\theoremstyle{definition}
\newtheorem{defn}{Definition}[section]
\theoremstyle{remark}
\newtheorem{rem}{Remark}[section]

  \newcommand{\oo}{\infty}
  \newcommand{\del}{\partial}
  \newcommand{\sse}{\subseteq}
  \newcommand{\sso}{\subset}
  \newcommand{\Alg}{\mathfrak{Alg}}
  \newcommand{\BkgG}{\mathfrak{BkgG}}
  \newcommand{\Bndl}{\mathfrak{Bndl}}
\renewcommand{\d}{\mathrm{d}}
  \newcommand{\D}{\mathcal{D}}
  \newcommand{\Diff}{\mathrm{Diff}}
  \newcommand{\E}{\mathcal{E}}
  \newcommand{\eps}{\varepsilon}
  \newcommand{\g}{\mathbf{g}}
  \newcommand{\h}{\mathbf{h}}
\renewcommand{\H}{\mathbf{H}}
  \newcommand{\id}{\mathrm{id}}
  
  \newcommand{\I}{\mathcal{I}}
  \newcommand{\LCV}{\mathfrak{LCV}}
  \newcommand{\Lie}{\mathcal{L}}
  \newcommand{\Man}{\mathfrak{Man}}
\renewcommand{\P}{\mathcal{P}}
  \newcommand{\R}{\mathbb{R}}
  \newcommand{\Riem}{\mathbf{R}}
\renewcommand{\S}{\mathbf{S}}
  \newcommand{\Secs}{\Gamma}
  \newcommand{\W}{\mathcal{W}}

\title{Analytic Dependence is an Unnecessary Requirement in
Renormalization of Locally Covariant QFT}

\author{Igor Khavkine\footnote{igor.khavkine@unitn.it} {} and
	Valter Moretti\footnote{valter.moretti@unitn.it} \\
\small 
	Dipartimento di Matematica,
	Universit\`a di Trento  and INFN-TIFPA Trento,\\
\small
	via Sommarive 14 I-38123 Povo (Trento), Italy}

\begin{document}
\maketitle

\begin{abstract}
Finite renormalization freedom in locally covariant quantum field
theories on curved spacetime is known to be tightly constrained, under
certain standard hypotheses, to the same terms as in flat spacetime up
to finitely many curvature dependent terms. These hypotheses include, in
particular, locality, covariance, scaling, microlocal regularity and
continuous and analytic dependence on the metric and coupling
parameters. The analytic dependence hypothesis is somewhat unnatural,
because it requires that locally covariant observables (which are
simultaneously defined on all spacetimes) depend continuously on an
arbitrary metric, with the dependence strengthened to analytic on
analytic metrics. Moreover the fact that analytic metrics are globally
rigid makes the implementation of this requirement at the level of local
$*$-algebras of observables rather technically cumbersome. We show that
the conditions of locality, covariance, scaling and a naturally
strengthened microlocal spectral condition, are actually sufficient to
constrain the allowed finite renormalizations equally strongly, thus
eliminating both the continuity and the somewhat unnatural analyticity
hypotheses. The key step in the proof uses the Peetre--Slov\'ak theorem on the
characterization of (in general non-linear) differential operators by their
locality and regularity properties.
\end{abstract}

\section{Introduction}\label{sec:intro}
Perturbative ultraviolet renormalization of locally covariant quantum
field theories in (globally hyperbolic) curved spacetime is a well
established  topic of algebraic quantum field theory, especially for
scalar fields  \cite{bfk,bf, hw-wick, hw2}. It essentially  deals with
two classes of objects: \emph{Wick polynomials} and \emph{time ordered
Wick polynomials}. Exactly as in flat spacetime, these objects can be
considered as the building blocks of the whole renormalization
procedure. Smeared versions of Wick polynomials, of their time ordered
products and of their derivatives generate an algebra $\W(M,\g)$, for a
given spacetime $(M,\g)$, enlarged in a controlled way from the algebra
of products of smeared linear fields. This enlarged algebra then
includes physically fundamental observables, such as the stress-energy
tensor, which is necessary, for instance, to evaluate the energy
densities and fluxes of physical processes in curved spacetimes like
particle creation or Hawking radiation. The stress-energy tensor is also
needed to compute the back reaction of the quantum matter on the
background geometry.

This paper deals only with Wick polynomials, or more precisely just Wick
powers with all results easily extended to all Wick polynomials by
linearity, though the presented results could in principle be adapted to
deal also with their derivatives and their time ordered products. In
curved spacetime, Wick polynomials have to satisfy stronger locality and
covariance requirements than in flat spacetime. These requirements are
conveniently stated in the language of category theory introduced
in~\cite{bfv}, which we also use here. We should stress, though, that
the categorical language primarily serves to compress somewhat long
lists of hypotheses into concise statements. Existence of locally
covariant Wick polynomials and their time ordered products was
established in the seminal works of Hollands and Wald, respectively in
\cite{hw-wick} and  \cite{hw2}. It is well known that, in flat
spacetime, time ordered Wick polynomials are not uniquely defined. This
fact survives the passage to curved spacetime.  However, unlike in flat
spacetime, the absence of a preferred reference state means that Wick
polynomials are \emph{themselves} not uniquely defined. The ambiguities
involved with the definition of these two classes of fields are
physically interpreted as \emph{finite renormalizations} or
\emph{renormalization counterterms}, upon adopting the natural locally
covariant generalization of Epstein--Glaser approach to renormalization.

Exactly as in flat spacetime, each fixed type of (either Wick or
time-ordered) polynomial admits a finite-dimensional class of
independent counterterms. In curved spacetime, this class is much larger
than in Minkowski space, because of the possible dependence of
counterterms on \emph{background curvature}. While this class may no
longer be finite-dimensional, it is still \emph{finitely generated} or
\emph{quasi-finite-dimensional} in a precise sense, because the
counterterms may depend only polynomially on the curvature scalars up to
a certain dimension. This remarkable result, in the case of Wick
polynomials, presented in \cite[Thm.~5.1]{hw-wick} and summarized before
the statement of our Theorem~\ref{thm:main}, is arrived at by imposing
severe constraints on Wick polynomials in addition to those of locality
and covariance. These requirements are of various kinds. Some arise from
heuristic properties of quantum free fields, e.g.,\ Hermiticity and
commutation relations. Other requirements concern microlocal features
which, loosely speaking, extend to curved spacetime the structure of
Fourier transforms of the relevant Green functions on Minkowski space.
Another requirement regards the behaviour of Wick polynomials under a
rescaling of the metric and the parameters $m^2$ and $\xi$ of the free
theory, which describe the field's mass%
	\footnote{As in \cite{hw-wick},  we will always treat $m^2$ as a real
	number, which could be either positive, zero, or even negative, as
	ultraviolet renormalization is not sensitive to the sign of $m^2$.} %
and its coupling to the curvature. Finally there are the technically
delicate requirements of \emph{continuous} and \emph{analytic}
dependence on the metric. The two latter requirements play a crucial
role in~\cite{hw-wick} in their proof of the strong restrictions on
possible finite renormalization counterterms that was mentioned above.

The main difficulty with defining a suitable notion of the
\emph{continuous} dependence of an element of the algebra $\W(M,\g)$ on
the metric $\g$ (and the other parameters $m^2$ and $\xi$) is that,
continuously changing the metric $\g \mapsto \g'$, the whole algebra
$\W(M,\g)$ changes correspondingly and algebras $\W(M,\g)$ and
$\W(M,\g')$ associated with different metrics are not canonically
isomorphic.  Therefore even just stating the condition of continuous
dependence  on $\g$ requires some finesse. Locality can be turned
into an advantage in this context~\cite{hw-wick}. One may restrict
attention to metric variations in a spacetime region $O\subset M$ with
compact closure. If $\g$ agrees with $\g'$ outside $O$, essentially
exploiting a suitable version of the \emph{time slice axiom}, it is
possible to naturally identify an element of $\W(M,\g)$ with a
corresponding element in $\W(M,\g')$, when both are supported in $O$.
Hence, a local version of the continuity requirement can be imposed by
means of this canonical identification.

The requirement of \emph{analytic} dependence is even trickier to state.
It is argued in \cite{hw-wick} that analytic dependence is necessary
because the remaining requirements would not be able to rule out the
undesirable infinite family of non-polynomial in curvature counterterms
that were considered in~\cite{tf}. There is an important subtle
technical issue that arises in stating this analytic dependence condition.
The way followed for stating the continuity dependence requirement in a
local region $O$ faces here  an insurmountable obstruction: analytic
metrics are \emph{rigid} and if they coincide outside $O$ they must
coincide also in $O$. The ingenious but cumbersome strategy elaborated
in~\cite{hw-wick} makes use of a special class of Hadamard states over
the considered algebras. Since no local analytic variations of the
metric are possible, they consider a joint analytic family $\g^{(s)}$ of
the metric on $O$ and a corresponding analytic family of quasifree
Hadamard states $\omega^{(s)}$ on $\W(M, \g^{(s)})$. Then they require
that the distributions obtained by composing $\omega^{(s)}$ with the
local Wick polynomials (or their time ordered products) varies
analytically with $s$ in a suitable \emph{analytic and microlocal} sense
(see the discussion starting on p.~311 in \cite{hw-wick}).   

Continuous and analytic dependence on the parameters $m^2$ and $\xi$ is
there treated similarly, with both parameters taken to be functions on
$M$, rather than just constants, at intermediate stages of the
arguments.

The main result of this work establishes that the technically cumbersome
and somewhat unnatural analytic dependence requirement is in fact
\emph{not necessary} to achieve the classification
theorem~\cite[Thm.~5.1]{hw-wick}. Our classification result,
Theorem~\ref{thm:main}, is essentially identical, though it is slightly
more general because it allows smooth (rather than just analytic)
dependence on the dimensionless curvature coupling $\xi$. In our proof,
we make no use of the continuous and analytic dependence requirements
of~\cite{hw-wick}. Instead, leaving all the other requirements on Wick
powers the same, we appeal to a strengthened (more precisely,
parametrized) version of the microlocal spectrum condition, which we
believe is natural from both the physical and geometrical points of
view. This modification of the axioms on Wick powers is sufficient to
achieve the desired classification. Also, echoing the original (rather
implicit) arguments of~\cite{hw-wick}, we believe that it is very likely
that our version of the axioms are satisfied by the standard locally
covariant Hadamard parametrix prescription for explicitly constructing
Wick powers, which would mean that our classification results are
non-vacuous. However, we leave a detailed verification of this claim to
future work.

The key tool exploited in our proof is a theorem that characterizes
(generally non-linear) differential operators in terms of their locality
and regularity properties. This theorem, known as the
\emph{Peetre--Slov\'ak theorem} (or sometimes the \emph{non-linear Peetre
theorem}), in its most elementary  version
(Proposition~\ref{prp:peetre}; see also Appendix~\ref{app:peetre-param}
for a more general statement) states the following: any map $D$, that
associates smooth sections $\psi\colon M \to E$ of a bundle $E\to M$ to
smooth sections $D[\psi]\colon M \to F$ of another bundle $F\to M$ in
such a way that $D[\psi](x)$ depends only on the germ of $\psi$ at $x$
for any point $x\in M$, is necessarily a differential operator of
locally bounded order, smoothly depending on its arguments and their
derivatives.  The $C_k$ coefficients that characterize renormalization
counterterms of Wick polynomials precisely map sections of the bundle of
metrics and parameters, $m^2$ and $\xi$, to scalar valued distributions
on a spacetime $M$. The microlocal conditions ensure that these
distributions are actually smooth functions, while the locality
requirement implies that the $C_k$ satisfy the hypotheses of
Peetre--Slov\'ak's theorem and hence must be differential operators. A
combination of the scaling and covariance requirements then shows that
the differential order of the $C_k$ is globally bounded and that their
dependence on the metric, $m^2$ and the derivatives of all the
parameters is polynomial, with coefficients smoothly depending on $\xi$.
Further, covariance also dictates that the derivatives of the metric
necessarily group into curvature scalars.

Notably, the analytic dependence requirement is not exploited in
establishing the above result. Within the context of our proof, counter
terms like $m^k F(R/m^2)$, where $R$ is the Ricci scalar and $F$ is any
smooth function with strong decay near $0$ and $\pm\oo$, as considered
in~\cite{tf}, are excluded because they violate the microlocal
requirement: there exists a choice of a spacetime $(M,\g)$ and of a
scalar field $m^2$ such that the counterterm is not smooth and hence has
non-empty wavefront set and the Wick polynomials modified by adding
these counterterms do not satisfy the microlocal requirement (neither
the original, nor our strengthened version).

This paper is organized as follows. Our main theorem and its proof are
presented in Sect.~\ref{sec:wick}. The proof is somewhat lengthy, but
straight forward. It relies on some preliminary definitions and results
discussed in Sect.~\ref{sec:geom}. In particular our basic version of
Peetre--Slov\'ak's theorem is stated in Sect.~\ref{sec:peetre} after a
quick summary of elementary facts about jet bundles in Sect.
\ref{sec:coords}, where we also introduce some useful coordinate
systems. Sect.~\ref{sec:scal} is devoted to introducing our notion of
scaling which is more precise but substantially equivalent to the one
employed in~\cite{hw-wick}. However, we are careful to identify two
different kinds of scalings (physical and coordinate), which were mixed
in~\cite{hw-wick} by the introduction of Riemann normal coordinates. The
remainder of Sect.~\ref{sec:geom} deals with notions and results,
especially on $GL(n)$ representation theory, which are useful for
imposing the covariance requirement. After recalling the definition and
properties of Wick polynomials, and the more general notion of locally
covariant quantum field, with the appropriate categorical language, we
state and prove our main result in several steps in
Sect.~\ref{sec:wick}. Sect.~\ref{sec:discuss} concludes the paper
with a discussion of the results and directions for future work.
Appendix~\ref{app:peetre-param} illustrates a more general version of
Peetre--Slov\'ak's theorem, which applies to differential operators with
parameters.

\section{Geometry of scaling and general covariance}\label{sec:geom}
In this section we discuss some aspects of the geometry of the higher
derivatives (jets) of metric and scalar fields under the action of
scaling and diffeomorphism transformations. These properties will be
crucial in the characterization of finite renormalizations in locally
covariant quantum field theory in Sect.~\ref{sec:wick}.

\subsection{Coordinates on jets}\label{sec:coords0}
In differential geometry, \emph{jets}~\cite{olver,kms} are a geometric
way of collecting information about higher derivatives of functions (or
bundle sections) on manifolds, similar to what the tangent and cotangent
bundles do for first derivatives. Jets have an invariant geometric
meaning even on manifolds without a preferred metric or connection.
Further, a choice of a coordinate chart on a manifold induces a choice
of adapted coordinates on the corresponding jet bundle. One advantage of
working with jets is that certain calculations are very conveniently
performed in such an adapted local coordinate chart, yet also lead to
global and geometrically invariant conclusions. Below, we briefly
discuss some variations on adapted local coordinate systems on the space
of jets of bundle of metrics with some scalar fields.

Consider a smooth map $f\colon \R^m \to \R^n$, such that $f(0) = 0$. The
\textbf{germ} of $f$ at $0\in \R^m$ is the equivalence class of smooth maps
$f'\colon \R^m \to \R^n$ that agree with $f$ on some neighborhood of
$0\in \R^m$. The \textbf{$r$-jet} of $f$ at $0\in\R^m$ is the equivalence
class of all smooth maps $f'\colon \R^m\to \R^n$ that have the same
Taylor expansion at $0$ as $f$ to order $r$, denoted $j^r_0 f$.
Obviously, the germ contains more information than a jet of any order.
These definitions are clearly local, both on the domain and the target
of a smooth map, and are invariant under $C^\oo$-changes of coordinates.
Thus, these definitions easily translate to maps between smooth
finite-dimensional (smooth) manifolds $M$, $N$ replacing $0\in \R^m$ and
$0\in \R^n$, respectively, by generic points $x\in M$, $y\in N$. In
particular, with the said $M$ and $N$, we denote by $J^r(M,N)$ the set
of all distinct jets $j^r_x f$ of all smooth maps $f\colon M\to N$ for
all $x\in M$. Also, if $E\to N$ is a smooth bundle over $N$, then we
denote by $J^rE$ or, for emphasis, by $J^r(E\to N)\sso J^r(N,E)$ the
subset of jets of smooth sections $f\colon N \to E$. Both $J^r(M,N)$ and
$J^r(E\to N)$ can be given structures of smooth manifolds. A fiber
$(J^rE)_x$ at $x\in N$ is diffeomorphic to $E_x\times \R^{s_r}$, where
$E_x$ is the fiber of $E$ and $s_r$ counts the components of all
(symmetrized)  partial derivatives up to order $r$. In fact, by
projection onto the target of each jet, $J^rE \to E\to N$ is an iterated
smooth bundle. Given a section $\psi\colon N\to E$, we can collect the
$r$-jets of $\psi$ over each point of $N$ into a section $j^r\psi \colon
N\to J^rE$ called the \textbf{$r$-jet extension} of $\psi$.

Let $(x^a,v^i)$ be a local adapted coordinate chart on a bundle $F\to
M$, where $(x^a)$ serve as coordinates on a domain $U\sse M$ and
$(x^a,v^i)$ serve as trivializing coordinates on the fibers of the
domain $V\sse F$ over $U$. For example, if $T^p_qM\to M$ is the bundle
of $(p,q)$-tensors we can choose coordinates $(x^a,t^{a_1\cdots
a_p}_{b_1\cdots b_q})$ on the projection pre-image $V$ of $U$, such that
a section $\tau\colon T^p_qM\to M$ could locally be written as
\begin{equation}
	\tau(x) = t^{a_1\cdots a_p}_{b_1\cdots b_q}(\tau(x)) \,
		\d{x}^{b_1}\cdots \d{x}^{b_q} \,
		\frac{\del}{\del x^{a_1}} \cdots \frac{\del}{\del x^{a_p}} .
\end{equation}
The local chart $(x^a,v^i)$ then induces an adapted coordinate system
$(x^a,v^i_A)$ on the domain $V^r\sse J^rE$ that is the projection
pre-image of $V$ and is diffeomorphic to $V^r\cong V\times \R^{s_r}$,
with $s_r$ as discussed above. Each $A = a_1\cdots a_l$, standing in for
an unordered (equivalently, fully symmetrized) collection of base
manifold coordinate indices, is a \textbf{multi-index} of size $|A|=l$,
with the range $l=0,1,\ldots, r$. The defining property of these
coordinates is the identity
\begin{equation}
	v^i_A(j^r \psi(x)) = \del_A v^i(\psi(x))
	= \frac{\del}{\del x^{a_1}} \cdots \frac{\del}{\del x^{a_l}} v^i(\psi(x)) ,
\end{equation}
for any section $\psi\colon M\to F$. Given such a coordinate system, for
brevity, we use the notation $\del_a = \del/\del x^a$ and $\del^A_i =
\del/\del v^i_A$ for corresponding coordinate vector fields.

\subsection{Coordinates on jets of metric and scalar fields}\label{sec:coords}
If $M$ is a $n$-dimensional smooth manifold,  let us now fix the bundle
$BM\to M$ given by the bundle product of the bundle $\mathring{S}^2T_*M$
of (smooth) Lorentzian metric $(0,2)$-tensors over $M$ and the trivial
bundle $\R\times M\to M$ of (smooth) scalar fields over $M$.  Let us
denote the sections of this bundle by $(\g,\xi) \colon M\to BM$. There
are several  local coordinate systems on $J^rBM$, of various merits,
which we discuss below.

\emph{Covariant coordinates.}
Given a local coordinate chart $(x^a)$ on $U\sse M$, we define the
corresponding adapted coordinates $(x^a,g_{ab},z)$ on $V\sse BM$, which
in turn induce the \textbf{covariant coordinates}
\begin{equation}
	(x^a,g_{ab,A},z_A) \quad \text{on $V^r\sse J^rBM$.}
\end{equation}
Notice that only $n(n+1)/2$ components of $g_{ab}$ take part in the
above coordinates, in view of the symmetry of the metric.

\emph{Contravariant coordinates.}
Recall that a Lorentzian metric $\g\colon M\to \mathring{S}^2T^*M$ is
invertible and hence defines a section $\g^{-1} \colon M\to
\mathring{S}^2TM$. The components of the inverse metric can be extracted
by functions $g^{ab}$ defined on all of $V\sse BM$, such that
$g^{ab}(\g^{-1}(x)) = g_{ab}(\g(x))$, which induce the functions
$g^{ab}_A$ on $V^r$ that satisfy $g^{ab}_A(j^r\g(x)) = \del_A
g^{ab}(\g(x))$. Then, using the notation $g^{AB} = g^{a_1 b_1} \cdots
g^{a_l b_l}$, for $|A|=|B|=l$, we define the following functions
\begin{align}
	g &= \left|\det g_{ab}\right| , &
	g^{ab,A} &= g^{AB} g^{ab}_B , &
	z^A &= g^{AB} z_A ,
\end{align}
where,  by invertibility of Lorentzian metrics, the function $g^{-1}$ is
well defined on all of $V^r$, since $g = \left|\det g_{ab}\right|$ is
never zero.  These functions make up the alternative set of local
\textbf{contravariant coordinates}
\begin{equation}
	(x^a, g^{ab,A}, z^A) \quad \text{on $V^r\sse J^rBM$,}
\end{equation}
with the caveat that as the set of functions $(g,g^{ab})$ is only
functionally independent up to the identity $g^{-1} = \left|\det
g^{ab}\right|$, for instance,  \emph{one of the contravariant
coordinates $g^{ab}$ can be replaced by $g$}. These coordinates have
convenient scaling properties that will be exploited in
Sect.~\ref{sec:scal}.

\emph{Rescaled  contravariant coordinates.}
Another  coordinate system that we introduce on $V^r\sse J^rBM$, the
\textbf{rescaled contravariant coordinates}, is a suitable rescaling of
the previous one. Namely, we introduce various factors of $g^\alpha$ in
the latter coordinates ($n$ being the dimension of $M$):
\begin{equation}\label{eq:gGS-coordsADDED}
	(x^a,g,
		g^{-\frac{1}{n}} g_{ab} ,
		g^{\frac{1}{n} + \frac{1}{n}|A|} g^{ab,A} ,
		g^{\frac{s}{2n} + \frac{1}{n}|A|} z^{A}) ,
\end{equation}
\emph{where one  of the $n(n+1)/2$  functions $g^{-\frac{1}{n}} g_{ab}$
is omitted and replaced by $g$}. This is because the functions
$g^{-\frac{1}{n}} g_{ab}$ are not functionally independent because of
the relation $| \det g^{-\frac{1}{n}} g^{ab}|=1$.

\emph{Curvature coordinates.}
Recall also that, given a Lorentzian metric $\g$, we can always define
the corresponding covariant derivative, or Levi-Civita connection,
$\nabla$ and the Riemann tensor $\Riem$. Using well known formulas, we
can define functions $\Gamma^a_{bc}$ and $\bar{R}_{abcd}$ on $V^r\sse
J^rBM$ that correspond to the coordinate components of the Christoffel
symbols and the fully covariant Riemann tensor. Define also the fully
contravariant tensor $\S$ with components
\begin{equation}
	\bar{S}^{abcd} = g^{aa'} g^{bb'} \bar{R}_{a'}{}^{(c}{}_{b'}{}^{d)}
		= g^{ab,cd} - g^{b(c,d)a} - g^{a(d,c)b} + g^{cd,ab} + \text{l.o.t} ,
\end{equation}
where l.o.t\ stands for terms that involve only coordinates of lower
derivative order. Finally, let $\Gamma^a_{bc,A}$ denote the components
of the coordinate $\del_A$ derivatives of $\Gamma^a_{bc}$, let
$\bar{S}^{abcd,A}$ denote the components of the symmetrized
contravariant $\nabla^{A} = \nabla^{(a_1} \cdots \nabla^{a_l)}$
derivatives of $\S$, and let $\bar{z}^A$ the components of the
symmetrized contravariant $\nabla^{A}$ derivatives of the scalar field
$\xi$. It is well-known~\cite{Iyer-Wald94,ta-sym}%
	\footnote{On page 490 of~\cite{ta-sym}, the unpublished
	report~\cite{at-spinors} is used as the main reference for the
	properties of these coordinates. In~\cite{torre-spinors} it is
	explained further that their origin goes back to at
	least~\cite{penrose-spinors} and even earlier to~\cite{thomas}.} %
that
\begin{equation}
	(x^a,g_{ab},\Gamma^a_{(bc,A)},\bar{S}^{ab(cd,A)},\bar{z}^{A})
\end{equation}
also defines a coordinate system on $V^r\sse J^rBM$, which we shall call
\textbf{curvature coordinates}.  Note that the barred coordinate
functions correspond to components of fully contravariant tensors. These
coordinate have convenient transformation properties under
diffeomorphisms that will be exploited in Sect.~\ref{sec:diff}.

\emph{Rescaled curvature coordinates.} The final coordinate system that we introduce on $V^r\sse J^rBM$, the \textbf{rescaled curvature coordinates}, merges some of the
properties of the systems $(x^a,
		g^{-\frac{1}{n}} g_{ab} ,
		g^{\frac{1}{n} + \frac{1}{n}|A|} g^{ab,A} ,
		g^{\frac{s}{2n} + \frac{1}{n}|A|} z^{A}) $ and $(x^a, g_{ab},
\Gamma^a_{(bc,A)}, \bar{S}^{ab(cd,A)}, \bar{z}^A)$. Namely, we again introduce
various factors of $g$ in the curvature  coordinates:
\begin{equation}\label{eq:gGS-coords}
	(x^a, g, 
	g^{-\frac{1}{n}} g_{ab}, \Gamma^a_{(bc,A)},
	g^{\frac{3}{n} + \frac{1}{n}|A|} \bar{S}^{ab(cd,A)},
	g^{\frac{s}{2n}+\frac{1}{n}|A|} \bar{z}^A) ,
\end{equation}
\emph{where, again,  one  of the $n(n+1)/2$  functions $g^{-\frac{1}{n}} g_{ab}$ is omitted and replaced by  $g$}.

\subsection{Locality and the Peetre--Slov\'ak theorem}\label{sec:peetre}
It is well known that \emph{linear} differential operators have the
property that they are \emph{support non-increasing}. The powerful,
original result of Peetre~\cite{Peetre1959,Peetre1960} shows that this
property is sufficient to characterize them in the context of $C^\oo$
differential geometry. A similar characterization holds even for {\em
non-linear} differential operators~\cite{slovak,kms,ns}, a version of which
we present below.

Before proceeding, we need a robust geometric notion of what a {\em
differential operator} is. Often, differential operators are defined by
their expressions in coordinate charts. Any such definition is
necessarily coordinate dependent and must be checked to agree on chart
overlaps. On the other hand, we can give a \emph{coordinate independent
and global} definition of differential operators using jets and the
$r$-jet extension map $j^r$ defined earlier in Sect.~\ref{sec:coords}.

Given a smooth bundle $E\to N$, recall that the $r$-jet extension acts
as a map $j^r\colon \Secs(E\to N) \to \Secs(J^rE\to N)$, where  as usual
$\Secs(G\to L)$ denotes the space of \emph{smooth sections} of the bundle
$G\to L$.  For our purposes, the map $j^r$ will serve as a universal
differential operator of order $r$ in the following sense.
\begin{defn}
Let  $E\to N$ and $F\to M$ be smooth bundles, and consider a map
$D\colon \Secs(E) \to \Secs(F)$.
\begin{itemize}
\item[\textbf{(a)}]  $D$ is a  \textbf{differential operator of globally
bounded order} if there exists an integer $r\ge 0$, the \textbf{order},
and a smooth function $d\colon J^r(E\to N) \to F$, considered as a
bundle map (i.e.,\ fiber preserving), such that for any section $\psi
\in \Secs(E)$ we have an associated section of the form $D[\psi] =
d\circ j^r\psi \in \Secs(F)$.
\item[\textbf{(b)}] $D$ is a \textbf{differential operator of locally
bounded order} if it satisfies a similar condition locally. Namely, for
any point of $y\in N$ and section $\phi\in \Secs(E)$, there exists a
neighborhood $U\sse N$ of $y$ with compact closure, together with an
integer $r\ge 0$, an open neighborhood $V^r \sse J^r(E\to N)$ of
$j^r\phi(U)$ projecting onto $U$, and a smooth function $d\colon V^r \to
F$ that respects the projections $V^r \to U$ and $F\to M$, such that
$D[\psi](x) = d\circ j^r\psi(x)$ for any $x\in U$ and any $\psi\in
\Secs(E)$ with $j^r\psi(U) \sso V^r$.
\end{itemize}
\end{defn}

If $E\to M$ and $F\to M$ are vector bundles over the same base manifold
$M$ and $D\colon \Secs(E)\to \Secs(F)$ is a \emph{linear} map such that
$\phi(x) = D[\psi](x)$ depends only on the germ of $\psi$ at $x\in M$
then it is clear that $D$ will be support non-increasing. Elementary
reasoning shows that a linear, support non-increasing map will also only
depend on germs.  So, another way to rephrase the Peetre theorem for
linear differential operators is as follows, where the dependence on the
germ replaces the support non-increasing property.

\begin{prop}[Linear Peetre's Theorem~{\cite{Peetre1959,Peetre1960}}]\label{prp:peetre0}
Let  $E\to M$ and $F\to M$ be vector bundles and $D\colon \Secs(E)\to
\Secs(F)$ a linear map such that $\phi(x) = D[\psi](x)$ depends only on
the germ of $\psi$ at $x\in M$. Then $D$ is a linear differential
operator of locally bounded order (with smooth coefficients in view of
the above definition).
\end{prop}

\noindent In other words, despite the fact
that germs potentially contain much more information that jets, such a
linear map that depends only on germs in fact sees only jets.

Phrased as above, in terms of germs, the hypotheses of Peetre's theorem
are immediately adaptable to the case when the map $D$ is non-linear and
acts on sections of (non-vector) smooth bundles. We will only require an
additional regularity%
	\footnote{In an earlier version of the manuscript, we mistakenly
	omitted the regularity hypothesis from the statement of the
	Peetre--Slov\'ak theorem. We thank the anonymous referee for bringing
	that to our attention.} %
hypothesis.
\begin{defn}\label{def:reg}
Given smooth bundles $E\to N$ and $F\to M$, a \textbf{smooth
($k$-dimensional) family} of sections of $E\to N$ is a smooth section of
the pullback bundle $\pi^* E\to \R^k\times N$
(cf.~Eq.~\eqref{eq:pbb-def}), where $\pi\colon \R^k\times N \to N$
is the projection onto the second factor, and similarly for families of
sections $F\to M$. A map $D\colon \Secs(E\to N) \to \Secs(F\to M)$ is
\textbf{regular} if it maps smooth families of sections to smooth
families of sections. A smooth family $\sigma \colon \R^k\times N \to
\pi^* E$ is called a \textbf{compactly supported variation} if there
exists a compact subset $O\sso N$ such that $\sigma$ is constant along
the $\R^k$ factor on the complement $\R^k\times N \setminus
\pi^{-1}(O)$. The map $D$ is \textbf{weakly regular} if it maps smooth
compactly supported variations to smooth compactly supported variations.
\end{defn}
\begin{prop}[Peetre--Slov\'ak's Theorem]\label{prp:peetre}
Let $E\to M$ and $F\to M$ be smooth bundles and $D\colon \Secs(E) \to
\Secs(F)$ a map such that $\phi(x) = D[\psi](x)$ depends only on the
germ of $\psi$ at $x\in M$. If in addition $D$ is weakly regular, then
it is a (non-linear) differential operator of locally bounded order.
\end{prop}
This proposition will be sufficient for our purposes. However, in the
standard literature~\cite{slovak}, \cite[\textsection~19]{kms}, this
result is stated in much greater generality. In fact, that level of
generality can obscure the meaning and significance of the theorem.
Though, it should be noted that a simplified statement of the theorem,
essentially identical to the one above, together with a straight-forward
self-contained proof recently appeared in~\cite{ns}. Note that these
standard references usually require \emph{regularity} instead of
\emph{weak regularity}, but a slight modification of the proof given
in~\cite{ns} makes it clear that only weak regularity is necessary. This
point is discussed in Appendix~\ref{app:peetre-param}.  Also in
Appendix~\ref{app:peetre-param}, we briefly introduce the language
needed to state a more general version,
Proposition~\ref{prp:peetre-param}. The above simpler version becomes a
special case of Proposition~\ref{prp:peetre-param} once it is trivially
checked that $D$ is $\id$-local, where $\id\colon M\cong M$ is the
identity map. The more general result given in
Appendix~\ref{app:peetre-param} serves two purposes. The first is that
it introduces the language in which the Peetre--Slov\'ak theorem and its
proof appear in the standard literature~\cite[\textsection 19]{kms},
which also refers to it as the \emph{non-linear Peetre theorem}.
Second, it allows the treatment of differential operators with
parameters. For instance, later in Sect.~\ref{sec:wick}, we treat the
mass $m^2$ of a scalar field and its coupling to curvature $\xi$ as
space-time dependent background fields. If they were treated as
necessarily spacetime-constant parameters, we would need to substitute
Proposition~\ref{prp:peetre-param} for the simpler
Proposition~\ref{prp:peetre} in the proof of our main
Theorem~\ref{thm:main}.

\subsection{Physical scaling}\label{sec:scal}
Referring to the already introduced bundle $BM \to M$, sections
$(\g,\xi)\in \Secs(BM)$ consist of a smooth Lorentzian metric $\g$ and a
smooth scalar field $\xi$ on $M$. We consider the following scaling
transformation $(\g,\xi) \mapsto (\lambda^{-2}\g, \lambda^s \xi)$ on
sections. We call this transformation a \textbf{physical scaling}, in
contrast to a different kind of scaling to be introduced in
Sect.~\ref{sec:diff}. We will need the following rather general
\emph{recursive} definition, where $\R^+ := (0, +\infty)$,

\begin{defn}\label{def:almost-hom}
Consider a linear representation of the multiplicative group $\R^+$ on a
vector space $W$, written as $W \ni F\mapsto F_\lambda \in W$, for every
$\lambda \in \R^+$. 
\begin{itemize}
\item[\textbf{(a)}] An element $F\in W$ is said to have
\textbf{homogeneous degree $k\in \R$} if
\begin{equation}
	F_\lambda = \lambda^k F\quad \text{for all $\lambda \in \R^+$}\:.
\end{equation}
\item[\textbf{(b)}] An element $F\in W$ is said to have \textbf{almost
homogeneous degree $k\in\R$ and order $l\in \mathbb{N}$} if $l\ge 0$ is
an integer such that (the sum over $j$ is omitted if $l=0$)
\begin{equation}\label{eq:almost-hom-def}
	F_\lambda = \lambda^k F
		+ \lambda^k \sum_{j=1}^{l} (\log^j \lambda) G_j ,
		\quad \text{for all $\lambda \in \R^+$},
\end{equation}
and for some $G_j\in W$ depending on $F$, which have respectively almost
homogeneous degree $k$ and order $l-j$.
\end{itemize}
The definition is recursive, with higher orders defined in terms of
lower ones. Clearly, an element that is almost homogeneous of order
$l=0$ is simply homogeneous.
\end{defn}

\begin{rem}
Besides \emph{almost homogeneous}, other common names found in the
literature include \emph{poly-homogeneous}, \emph{associated
homogeneous} and even \emph{quasi associated homogeneous}. We are mostly
interested in the case when $W$ is some function space and the action of
$\R^+$ is induced from an action on the domain of the functions.
Reference~\cite{shelkovich} reviews several definitions leading to this
class of functions and lists relevant earlier works. In the context of
distribution theory, the terminology of \emph{associated homogeneous} is
prevalent and goes back to the seminal
references~\cite[\textsection~1.4]{gs1} and~\cite[Ch.I~\textsection
4]{gs2}. Our Definition~\ref{def:almost-hom} coincides
with~\cite[Def.~5.2]{shelkovich}.
\end{rem}

The \textbf{physical scaling transformation} on the sections $\Secs(BM)$
can be implemented by post-composing a section with a bundle map $BM\to
BM$:
\begin{equation}
	BM \ni (p, \g(p), z(p))
		\mapsto (p, \lambda^{-2}\g(p), \lambda^s z(p)) \in BM\:,
\end{equation}
where the real $\lambda \in \R^+$ defines the scaling transformation.
This representation of the multiplicative group $\R^+$  is
\emph{globally} defined,  however  this global action can be written in
adapted local coordinates, as discussed in Sect.~\ref{sec:coords}, and
looks like
\begin{equation}
	x^a \mapsto x^a, \quad g_{ab} \mapsto \lambda^{-2} g_{ab}, \quad
	z \mapsto \lambda^s z .
\end{equation}
This global transformation lifts to a global transformation of the jet
bundle $J^rBM$.  In the corresponding induced local coordinates, the
lifted action reads
\begin{equation}
	g_{ab,A} \mapsto \lambda^{-2} g_{ab,A}, \quad
	z_A \mapsto \lambda^s z_A .
\end{equation}
We are interested in applying Definition~\ref{def:almost-hom} to $W =
C^\oo(J^rBM)$ and the $\R^+$ action induced by the lift of physical
scalings to $J^rBM$.  Moreover, we will need to consider also smaller
domains $V^r\sse J^rBM$ for these functions, with $V^r$ themselves not
invariant under physical scalings. Thus, it is more convenient to refer
to the infinitesimal version of these transformations, which are
effected by the following vector field
\begin{equation}\label{e}
	e = -2 g_{ab,A} \del^{ab,A} + s z_A \del^A_z ,
\end{equation}
in the sense that the induced action on scalar functions on $J^rBM$
satisfies
\begin{equation}\label{eq:euler-vf}
	\left. \frac{d}{d\lambda} \right|_{\lambda = 1} F_\lambda
		= \Lie_e F .
\end{equation}
(In the rest of the paper  if $X$ is a vector field on $J^rBM$, $\Lie_X$
denotes the standard Lie derivative so that, in particular  $\Lie_X(F)
:= X(F)$ if  $F\colon V^r \sse J^rBM \to \R$ is a smooth function.)
Notice that, as the  physical scaling transformation is globally
defined, $e$ turns out to be globally defined  on  $J^rBM$ and~\eqref{e}
is just its expression in local coordinates. We have a first elementary
result stated within the following lemma.  We will essentially show
later that the converse implication holds as well.

\begin{lem}\label{lem:inf-almost-hom}
A smooth function $F\colon J^rBM \to \R$ that has almost homogeneous
degree $k$ and order $l$, according to Definition~\ref{def:almost-hom},
when the action $F \to F_\lambda$ is the one induced by physical scaling
transformations, satisfies the following local infinitesimal version
\begin{equation}
	(\Lie_e - k)^{l+1} F = 0 \:.
\end{equation}
\end{lem}

\begin{proof}
It is sufficient to make use of equation~\eqref{eq:euler-vf} and recall
the obvious identity $(\lambda d/d\lambda - k)^{l+1} \lambda^k \log^l
\lambda = 0$.
\end{proof}

This lemma is essentially a restatement of Theorem~5.2 and Remark~5.1
from~\cite{shelkovich}. It now permits us to give a definition of almost
homogeneity under infinitesimal scaling for functions defined on subsets
of jets of dimensionful bundles. The advantage of using infinitesimal
scaling is that the domain on which it is defined need not actually be
invariant under finite scaling.
\begin{defn}\label{def:almost-hom-inf}
A smooth function $F\colon V^r \sse J^rBM \to \R$, where $V^r$  is an
open subset which may coincide with all of $J^rBM$,  is said to have
\textbf{almost homogeneous degree $k\in \R$ and order $l\in \mathbb{N}$}
(with $l\ge 0$) \textbf{under physical scalings} if it satisfies the
identity
\begin{equation}
	(\Lie_e - k)^{l+1} F = 0 .
\end{equation}
If $l=0$, $F$ is said to have \textbf{homogeneous  degree $k\in\R$}.
\end{defn}
To investigate the local structure of $F$ above we initially use an open
subset $V^r$  equipped with  the contravariant coordinates
$(x^a,g^{ab,A},z^{A})$ introduced in Sect.~\ref{sec:coords}. In these
coordinates, finite and infinitesimal physical scalings take the form
\begin{gather}
	x^a \mapsto x^a, \quad
	g \mapsto \lambda^{-2n} g, \quad
	g^{ab,A} \mapsto \lambda^{2+2|A|} g^{ab,A}, \quad
	z^{A} \mapsto \lambda^{s+2|A|} z^{A} , \\
\label{eq:scal1-vf}
	e = (2+2|A|) g^{ab,A} \del_{ab,A} + (s+2|A|) z^A \del^z_A \:,
\end{gather}
where we have also described the action of rescaling on $g$ which, as
already remarked, can be used as an alternative coordinate in place of
one of the $g^{ab}$.  As $e$ does not  vanish anywhere, $J^rBM$ and
hence the domain $V^r$ are foliated by integral curves of the vector
field $e$. Moreover, the identity  $\Lie_e g^{-\frac{1}{2n}} =
g^{-\frac{1}{2n}}$ means that $g$ restricts to a global coordinate on
each orbit of $e$. Thus, the level sets of $g$ constitute another
foliation of $J^rBM$ and $V^r$, transverse to the integral curves of
$e$. These observations suggest to study the structure of (almost)
homogeneous functions of degree $k$ in the \emph{rescaled contravariant
coordinates} 
\begin{equation}
	(x^a, g, g^{-\frac{1}{n}} g_{ab} ,
		g^{\frac{1}{n} + \frac{1}{n}|A|} g^{ab,A} ,
		g^{\frac{s}{2n} + \frac{1}{n}|A|} z^{A}) ,
\end{equation}
that were introduced in Sect.~\ref{sec:coords}. Note that \emph{each
of these functions but $g$ is invariant under physical scalings}. We
have the following result.
\begin{lem}\label{lem:almost-hom}
Suppose that $V^r \sse J^rBM$ is an open set equipped with either
coordinates $(x^a,g^{ab,A},z^{A})$ or some other coordinate system
introduced in Sect.~\ref{sec:coords}, and $F\colon  V^r  \to \R$ is a
smooth function that has almost homogeneous degree $k$ and order $l$
with respect to physical scalings, as in
Definition~\ref{def:almost-hom-inf}. Then there exist  homogeneous of
degree $0$  functions $H_j\colon V^r \to \R$, for $j=0,1,\ldots, l$,
such that
\begin{equation}\label{eq:almost-hom}
	F = g^{-\frac{k}{2n}} \sum_{j=0}^{l}
		\log^j (g^{-\frac{1}{2n}}) H_j .
\end{equation}
In particular, using rescaled contravariant coordinates, each $H_j$ can
be taken independent of $g$ and written in the form
\begin{equation}\label{eq:det-indep}
	H_j = H_j(x^a,
		g^{-\frac{1}{n}} g_{ab} ,
		g^{\frac{1}{n} + \frac{1}{n}|A|} g^{ab,A} ,
		g^{\frac{s}{2n} + \frac{1}{n}|A|} z^{A}) .
\end{equation}
\end{lem}

\begin{proof}
In the simplest $l=0$ case, we can define $H = g^{\frac{k}{2n}} F$ and
show that $\Lie_e H = 0$ because $\Lie_e g^{-\frac{1}{2n}} =
g^{-\frac{1}{2n}}$. This means that, in  rescaled contravariant
coordinates, $H$ is independent of $g$ and hence \eqref{eq:det-indep}
holds, with $H$ in place of $H_j$. Next, the general $l\ge 1$ case can
be treated as follows. Let $G := g^{\frac{k}{2n}} F$, which implies that
$\Lie_e^{l+1} G = g^{\frac{k}{2n}} (\Lie_e-k)^{l+1} F = 0$. Now, note the identity $\Lie_e^j \log^j
(g^{-\frac{1}{2n}}) = j!$. So, if $H_{l} := \frac{1}{l!} \Lie_e^{l} G$
and $G_{l-1} := G - \log^{l} (g^{-\frac{1}{2n}}) H_{l}$, then $\Lie_e
H_{l} = 0$ and $\Lie_e^{l} G_{l-1} = 0$. In other words, starting with
$G_{l} = G$, we can recursively define $H_j := \frac{1}{j!} \log^j
(g^{-\frac{1}{2n}}) \Lie_e^j G_j$ and $G_{j-1} := G_{j} - \log^{j}
(g^{-\frac{1}{2n}}) H_{j}$, finding  $\Lie_e H_{j} = 0$ at each step.
The procedure stops for $j=0$ when it gives $G_{0}=H_{0}$, so that
$G_{j<0} = H_{j<0} = 0$, proving \eqref{eq:almost-hom}.
\end{proof}

We will also need the following basic result regarding \emph{products}
of vectors with almost homogeneous degree as in
Definition~\ref{def:almost-hom}. Due to the generality of
Definition~\ref{def:almost-hom} we must clarify the meaning of
\emph{product}. If $W$ and $W'$ are two vector spaces,  by a
\textbf{product} between them,  we mean any fixed bilinear map $W\times
W' \to V$, where $V$ is another vector space. If $F\in W$ and $F'\in W'$
the corresponding element in $V$, their \textbf{product},  will be
simply denoted by $FF' \in V$.

\begin{lem}\label{lem:almost-hom-alg}
Referring to Definition \ref{def:almost-hom}, consider a pair of vector
spaces $W, W'$ endowed with corresponding representations of $\R^+$.
Concerning (b) below, assume also that there is a product  $W\times W'
\to V$ such that  (i) $V$ admits a representation of $\R^+$ and (ii) the
map  $W\times W' \to V$ is equivariant: $F_\lambda F'_\lambda =
(FF')_\lambda$ for $F\in W$, $F'\in W'$ and $\lambda \in \R^+$.
\begin{itemize}
\item[\textbf{(a)}] A linear combination of two elements $F, F' \in W$
of almost homogeneous degree $k$ and order $l$ is of almost homogeneous
degree $k$ and order $l$.
\item[\textbf{(b)}]  A product of an element $F\in W$, of almost
homogeneous degree $k$ and order $l$, and an element $F'\in W'$, of
almost homogeneous degree $k'$ and order $l'$, has almost homogeneous
degree $k+k'$ and order $l+l'$.
\end{itemize}
\end{lem}

\begin{proof}
Part (a) is trivial, because the defining
identity~\ref{eq:almost-hom-def} is linear.

We will prove part (b) by double induction on the pair of orders
$(l,l')$. Consider the identity
\begin{multline}\label{eq:almost-hom-trans}
	(F F')_\lambda = F_\lambda F'_\lambda
		= \lambda^{k+k'} F F' \\
			+ \lambda^{k+k'} \sum_{j=1}^l (\log^j \lambda) G_j F'
			+ \lambda^{k+k'} \sum_{j'=1}^{l'} (\log^{j'} \lambda) F G'_{j'} \\
			+ \lambda^{k+k'} \sum_{j=1}^l \sum_{j'=1}^{l'}
					(\log^{j+j'}\lambda) G_j G'_{j'} .
\end{multline}
From this formula, it is clear that, to show that $F F'$ has almost
homogeneous degree $k+k'$ and order $l+l'$, it is sufficient to
establish that the coefficients of the logarithmic terms, $G_j F'$, $F
G'_{j'}$ and $G_j G'_{j'}$, either do not appear or are themselves
almost homogeneous of the right degree and order. Thus, to establish the
case $(l,l')$, it is sufficient to have all of the $(j,l')$, $(l,j')$
and $(j,j')$ cases, with $j<l$ and $j'<l'$, already established. We
shall refer to this last remark as the \emph{primary inductive step}.

The case $(l,l') = (0,0)$ follows immediately from
Eq.~\eqref{eq:almost-hom-trans}, since no logarithmic terms appear.
Next, we establish the following \emph{secondary inductive step}.
Assuming that, given some $m\ge 0$, all cases $(l,l')$ with $l,l'\le m$
hold, then actually all cases $(l,l')$ with $l,l'\le m+1$ hold as well.
To see that, note that the case $(m+1,0)$ holds, because
in~\eqref{eq:almost-hom-trans} we need only consider the terms $G_j F'$,
which correspond to the inductively covered cases $(m+1-j,0)$ with $j\ge
1$. Then, using the primary inductive step, all the cases $(m+1,l')$
with $1\le l'\le m$ follow as well. The cases $(l,m+1)$ with $0\le l\le
m$, are completely analogous. Finally, one more appeal to the primary
inductive step establishes the case $(m+1,m+1)$.

Iterating the secondary inductive step completes the proof of part (b).
\end{proof}

\subsection{Diffeomorphisms and coordinate scalings}\label{sec:diff}
Because the sections $(\g,\xi) \in \Secs(BM)$ are tensor fields, there
is a well defined action of the group $\Diff(M)$ of diffeomorphisms
$\chi\colon M\to M$ on them by pullback $(\g,\xi) \mapsto (\chi^*\g,
\chi^*\xi)$. This action of course can be implemented at the level of
the bundle itself, $\chi^*\colon BM \to BM$ and of course lifted to the
jet bundle $j^r\chi^* \colon J^rBM \to J^rBM$. We are interested in the
structure of functions $F\colon J^rBM \to \R$ that are invariant under
the action of $\Diff(M)$. We could also consider invariance only under
the subgroup $\Diff^+(M)$ of orientation preserving diffeomorphisms
 in an essentially analogous way. For this purpose, it is
convenient to make use of the local adapted \emph{curvature coordinates}
$(x^a,g_{ab},\Gamma^a_{(bc,A)}, \bar{S}^{ab(cd,A)}, \bar{z}^A)$ on a
domain $V^r\sse J^rBM$ defined in Sect.~\ref{sec:coords}.

The domain $V^r$ itself may not be invariant under $\Diff(M)$, because
our coordinates are adapted to a single coordinate chart $(x^a)$ on
$U\sse M$. On the other hand, having already chosen our coordinate
system, we can phrase the requirement that $F\colon V^r\to \R$ is the
restriction of a $\Diff(M)$-invariant function (necessarily defined on a
possibly larger $\Diff(M)$-invariant domain) to $V^r$ in the following
way: (a) $\frac{\del}{\del x^a} F = 0$, where the vector fields
$\frac{\del}{\del x^a}$ are the infinitesimal generators of
diffeomorphisms that restrict to coordinate translations on $U$, and (b)
the restriction $F_x \colon V^r_x \sse J^r_xBM \to \R$ of $F$ to the
fiber of $J^rBM$ over any one point $x\in M$ is invariant under the
action of the subgroup $\Diff(M,x) \sso \Diff(M)$ that fixes $x$.
Clearly we can take $V^r_x$ to be invariant under $\Diff(M,x)$. An
immediate simplification based on requirement (a) is that our function
is expressible as $F = F_x(g_{ab}, \Gamma^a_{(bc,A)},
\bar{S}^{ab(cd,A)}, \bar{z}^{A})$, that is, it is independent of the
base coordinates $(x^a)$. Next, we examine the consequences of
requirement (b).

The action of $\Diff(M,x)$ on $r$-jets is not faithful. In fact, it has
a large kernel, so that the action on $J_x^rBM$ factors through the
homomorphic projection $\Diff(M,x) \to G^r_n$, where $G^r_n$ is a
finite-dimensional Lie group known as the \textbf{$r$-jet
group}~\cite[\textsection~13]{kms}. Thus, we need only consider the
invariance of $F_x$ under $G^r_n$. The $r$-jet groups come with natural
projections $G^r_n \to G^{r-1}_n$, corresponding to the equivariant
projection $J_x^rBM \to J_x^{r-1}BM$, and it is easily seen that $G^1_n
\cong GL(n)$. Analogously, for orientation preserving diffeomorphisms,
we denote the corresponding projections as $\Diff^+(M) \to G^{+r}_n \to
GL^+(n)$.

The curvature coordinates $(g_{ab}, \Gamma^a_{(bc,A)},
\bar{S}^{ab(cd,A)}, \bar{z}^A)$ are used specifically for their transformation
properties under $G^r_n$. Note that, without loss of generality but
after a possible small restriction of $V_x^r$, we can factor $V_x^r
\cong \R^\gamma \times W^r$, where the projection onto the $\R^\gamma$
factor is effected by the $(\Gamma^a_{(bc,A)})$ coordinates and the
projection onto the $W^r$ factor is effected by the remaining
coordinates. This factorization respects the action of $G^r_n$ in the
sense that the projection $V_x^r \to W^r$ induces a well-defined action
of $G^r_n$ and $W^r$. The action on $W^r$ actually factors
through the projection $G^r_n \to G^1_n \cong GL(n)$, since it is
coordinatized by components of tensors. Moreover, for any $w\in W^r$,
the isotropy subgroup of $w$ in $G^r_n$ acts transitively on the fiber
$\R^\gamma$ over $w$. In the orientation preserving case, the same is
true of the corresponding actions of $G^{+r}_n$ and $GL^+(n)$. The fact
that $G^r_n$ (and also $G^{+r}_n$) acts transitively on the $\R^\gamma$
fibers that are coordinatized by the derivatives of the Christoffel
symbols $(\Gamma^a_{(bc,A)})$ means that an invariant function $F_x$
cannot depend on these coordinates, which is a well-known result that is
sometimes known as the \emph{Thomas replacement
theorem}~\cite{Iyer-Wald94,ta-sym}. Let us rephrase it slightly below.

The above factorization $V^r \cong \R^\gamma \times W^r$ is also
compatible with the rescaled curvature coordinates
\begin{equation} \label{eq:gGgS-coords}
	(x^a,g,
	g^{-\frac{1}{n}} g_{ab}, \Gamma^a_{(bc,A)},
	g^{\frac{3}{n} + \frac{1}{n}|A|} \bar{S}^{ab(cd,A)},
	g^{\frac{s}{2n}+\frac{1}{n}|A|} \bar{z}^A) ,
\end{equation}
that were introduced in Sect.~\ref{sec:coords}. Recall that in our
notation the functions $(g^{-\frac{1}{n}} g_{ab})$ are functionally
independent only up to the identity $| \det (g^{-\frac{1}{n}} g_{ab})| =
1$. The main distinction is that these coordinates, other than $(x^a,
\Gamma^a_{(bc,A)})$, are no longer components of tensors, but rather of
tensor densities, which also transform under $GL(n)$
(cf.~Sect.~\ref{sec:iso-tens}). Using these coordinates, together with
the preceding discussion, we can simplify a $\Diff(M)$-invariant $F$ as
follows:
\begin{prop}[Thomas replacement theorem]\label{prp:thomas}
Let $F\colon V_x^{\prime r}\sse J^rBM \to \R$ be a $\Diff(M)$-invariant
function defined on a $\Diff(M)$-invariant domain. In the coordinate
system~\eqref{eq:gGgS-coords} defined on the domain $V^r\sse V^{\prime r}$, the
restriction of $F$ to $V^r$ must be expressible as
\begin{equation}
	F = G(g, g^{-\frac{1}{n}} g_{ab},
		g^{\frac{3}{n} + \frac{1}{n}|A|} \bar{S}^{ab(cd,A)},
		g^{\frac{s}{2n}+\frac{1}{n}|A|} \bar{z}^A) ,
\end{equation}
where the function $G$ is invariant under the action of $GL(n)$ on its
arguments.
\end{prop}

At this point, we have reduced the invariance of $F$ under $\Diff(M)$ to
the invariance of the function $G$, from Proposition~\ref{prp:thomas},
under  $GL(n)$ (obtained as the projection $\Diff(M,x) \to GL(n)$),
which follows from the preceding discussion. Analogous statements hold
for $\Diff^+(M)$, $\Diff^+(M,x)$ and $GL^+(n)$. We now single out a
specific subgroup of $GL^+(n)$ (and hence also of $GL(n)$) that we shall
call the group of \textbf{coordinate scalings}. It consists of matrices
of the form $\mu I_n \in GL(n)$, where $\mu$ is a positive real number
and $I_n$ is the $n\times n$ identity matrix. The name refers to the
fact that $\mu I_n$ is the image of a diffeomorphism that restricts to a
uniform scaling of the coordinates $(x^a)$ centered at $x\in U\sse M$,
with of course many other possible pre-images, under the projection
$\Diff^+(M) \to GL^+(n)$. These transformations should be contrasted
with the distinct group of \emph{physical scalings} introduced in
Sect.~\ref{sec:scal}.

Coordinate scalings act on the components of tensor densities appearing
in the coordinate system~\eqref{eq:gGgS-coords} as follows:
\begin{align}
	g &\mapsto \mu^{2n} g , &
	g^{\frac{3}{n}+\frac{1}{n}|A|} \bar{S}^{ab(cd,A)}
		&\mapsto \mu^{2+|A|} g^{\frac{3}{n}+\frac{1}{n}|A|} \bar{S}^{ab(cd,A)} , \\
	g^{-\frac{1}{n}} g_{ab} &\mapsto g^{-\frac{1}{n}} g_{ab} , &
	g^{\frac{s}{2n} + \frac{1}{n}|A|} \bar{z}^A
		&\mapsto \mu^{s+|A|} g^{\frac{s}{2n} + \frac{1}{n}|A|} \bar{z}^A .
\end{align}
We stress a fundamental difference between \emph{coordinate scalings}
and the previously introduced \emph{physical scalings}: coordinate
scalings are induced from the action of the diffeomorphism group, while
the physical ones are not.

\subsection{Equivariant and isotropic tensors}\label{sec:iso-tens}
In this section, we present some basic facts about \emph{equivariant} maps
between spaces that carry certain representations of $GL(n)$.

In particular, consider the space $B_n$ of bilinear forms on $\R^n$, and
the natural linear action of $GL(n)$ thereon. The subset $L_n \sso B_n$
of non-degenerate bilinear forms of Lorentzian signature $(-+{\cdots}+)$
is invariant and hence inherits an action of $GL(n)$ itself.  If $\eta
\in L_n$ is the \textbf{canonical Lorentzian form}, defined by the matrix
$\mathrm{diag}(-1,1,\ldots, 1)$ referring to the canonical basis of
$\R^n$, the subgroup $O(1,n-1)\subset GL(n)$ is defined as the {\em
isotropy group} of $\eta$.  We could also restrict the action on $L_n$
to the subgroup $GL^+(n) \sso GL(n)$ of \emph{orientation preserving}
transformations.  With this choice,  the isotropy group of $\eta$
turns out to be $SO(1,n-1)= O(1,n-1) \cap GL^+(n)$.

\begin{rem}
$L_n$ consists of a single orbit and is in fact isomorphic to the
homogeneous space $GL(n)/O(1,n-1)$. Similarly, $L_n$ is also isomorphic
to the homogeneous space $GL^+(n)/SO(1,n-1)$. The fact that the action
of $GL(n)$ (resp.~$GL^+(n)$) is transitive on $L_n$ implies, as a
general well-known fact, that the isotropy group of any  $g\in L_n$ is
isomorphic to $O(1,n-1)$ (resp.~$SO(1,n-1)$).
\end{rem}

\begin{defn}\label{defT}
Let $M^p_n$ be the space of $p$-multilinear forms on $\R^n$ and consider
the natural linear action of $GL(n)$ thereon. Let $T$ be a
finite-dimensional real vector space carrying a representation of
$GL(n)$.
\begin{itemize}
\item[\textbf{(a)}] $T$ is a \textbf{(covariant) tensor representation}
if it is the restriction of the action of $GL(n)$ on $M^p_n$ with
respect to some linear embedding $T\hookrightarrow M^p_n$ as an
invariant subspace. We call $p$ the tensor \textbf{rank} of $T$.
\item[\textbf{(b)}] $T$ a \textbf{(covariant) tensor density
representation} if $T$ is as in (a) but the action of $GL(n) \ni u
\mapsto \rho(u)$ on $T$ is given by a tensor representation up to a
multiplication by $\left|\det u\right|^s$, where $s$ is the tensor
\textbf{weight} of $T$.
\end{itemize}
\end{defn}

Of course, we obtain similar definitions by substituting $GL^+(n)$ for
$GL(n)$, and also $O(1,n-1)$ or $SO(1,n-1)$, when a particular
Lorentzian bilinear form $g$  is fixed. Of course, in the case of
$O(1,n-1)$ and $SO(1,n-1)$, there is no distinction between
\emph{tensor} and \emph{tensor density} representations.

Finally, it is useful to consider the one point space ${*} \cong \R^0$
with the trivial action of $GL(n)$ or any of its subgroups thereon.

\begin{defn}
Let $X$ and $Y$ be spaces carrying respective actions $\rho^{(X)}$ and
$\rho^{(Y)}$ of the group $G$. A map $f\colon X \to Y$ is  said to be
\textbf{equivariant} if it commutes with the action of $G$:
\begin{equation}
	f\circ\rho^{(X)}_u  = \rho^{(Y)}_u \circ f
	\quad\text{for every $u\in G$}\:.
\end{equation}
Consider the special case where  $X:={*}$, $Y:= T$ as in (a) in
Definition \ref{defT}, and $G:= O(1,n-1)$. The image of an equivariant
map $* \to T$  is called an \textbf{$O(1,n-1)$-isotropic tensor}.  The
space of $O(1,n-1)$-isotropic tensors in $T$ will be denoted by
${\I}_T$.

An \textbf{$SO(1,n-1)$-isotropic tensor} is defined similarly, replacing
$O(1,n-1)$ by $SO(1,n-1)$ everywhere. The space of $SO(1,n-1)$-isotropic
tensors in $T$ will be denoted by $\tilde{\I}_T$.
\end{defn}

\begin{rem} \strut

\begin{itemize}
\item[\textbf{(1)}] The embedding  $T\hookrightarrow M^p_n$ is an
evident example of equivariant map for $GL(n)$ (and every subgroup)  by
definition.
\item[\textbf{(2)}]  As $f\colon * \to T$ is completely defined by its
image $f(*) = t \in T$  the definition states that a tensor $t \in T$ is
isotropic if it is \emph{invariant} under the relevant action of
$O(1,n-1)$ (or $SO(1,n-1)$) on $T$.
\item[\textbf{(3)}] The space of isotropic tensors for different
Lorentzian bilinear forms are  clearly isomorphic.
\end{itemize}
\end{rem}

It is well known that the subspaces of isotropic tensors $\I_T \sso T$
and  $\tilde{\I}_T \sso T$ can be fully characterized as in the
proposition below. In the following, $\epsilon \in M^n_n$ denotes the
\textbf{canonical Levi-Civita tensor}, that is, the  fully
anti-symmetric form uniquely fixed by the value of its component
$\epsilon_{1\cdots n} = 1$, with respect to  the canonical basis of
$\R^n$. Also, $\I^p_{n} \sso M^p_n$ will denote the subspace spanned by
all possible tensor products of the canonical Lorentzian form $\eta \in
L_n$ that create a $p$-multilinear form. More precisely, $\I^p_n$ is
spanned by elements of the form
\begin{equation}
	(\eta_\sigma)_{i_1 i_2 \cdots i_{p-1} i_p}
		= \eta_{\sigma(i_1)\sigma(i_2)}
			\cdots \eta_{\sigma(i_{p-1})\sigma(i_p)} ,
\end{equation}
where $\sigma\in S_p$ is any permutation. Similarly, $\tilde{\I}^p_{n}
\sso M^p_n$ denotes the subspace spanned by all possible tensor products
of $\eta$ and $\epsilon$ that create a $p$-multilinear form.

\begin{prop}\label{prp:iso-tens}
Given a real vector space $T$ carrying a tensor representation of
$GL(n)$ and identifying $T$ with its image with respect to the embedding
$\alpha: T\hookrightarrow  M^p_n$, the following facts hold.\\
\textbf{(a)} The subspace  $\I_T \sso T$ is given by $\I_T \cong
\alpha(T)\cap \I^p_{n}$.\\
\textbf{(b)} The subspace  $\tilde{\I}_{T} \sso T$ is given by
$\tilde{\I}_{T} \cong \alpha(T) \cap \tilde{\I}^p_{n}$.
\end{prop}
An elementary proof of such a characterization of $O(n)$- and
$SO(n)$-isotropic tensors can be found in~\cite{ado}, which generalizes
straightforwardly to $O(1,n-1)$ and $SO(1,n-1)$. More generally, this
kind of result is sometimes known as \emph{first fundamental theorem} of
\emph{invariant theory}~\cite{weyl,gw} for the corresponding group.

\begin{defn}\label{def:eqtd}
Given a real vector space $T$ with a tensor density representation of
$GL(n)$ (resp.~$GL^+(n)$) and the natural representation on $L_n$, we
will refer to an equivariant map $t: L_n \to T$ as a
\textbf{$GL(n)$-equivariant tensor density}, and similarly for
\textbf{$GL^+(n)$-equivariant tensor densities}.

The space of $GL(n)$-equivariant tensor densities will be denoted by
$\E_T$ and the space of $GL^+(n)$-equivariant tensor densities will  be
denoted by $\tilde{\E}_T$.
\end{defn}

\begin{rem}
Even if the functions belonging to $\E_T$ and $\tilde{\E}_T$ are not
required to be linear, these spaces enjoy  a natural structure of  {\em
real vector space}, just in view of the fact that the equivariant tensor
densities are maps with values in the real vector space $T$.
\end{rem}

The following lemma characterizes the space of equivariant tensor
densities (in the sense of equivariant maps) in terms of isotropic
tensors (in the sense of the subspaces $\I_T \sse T$
(resp.~$\tilde{\I}_T \sse T$) defined earlier).

\begin{lem}\label{lem:equiv-tens}
Let $T$ be a finite-dimensional real vector space carrying a tensor
density representation of $GL(n)$, resp.~$GL^+(n)$, and assume that
$L_n$ is equipped with the natural representation.

\textbf{(a)} The space of $GL(n)$-equivariant, resp.~$GL^+(n)$-equivariant,
tensor densities is isomorphic the subspace of $O(1,n-1)$-isotropic
tensors, resp.~$SO(1,n-1)$-isotropic tensors, in $T$. More precisely,
the isomorphism is defined by
\begin{equation}
	\E_T  \ni t \mapsto t(\eta) \in \I_{T}
	\quad
	(\text{resp.} \quad
		\tilde{\E}_{T}  \ni t \mapsto t(\eta) \in \tilde{\I}_{T}).
\end{equation}

\textbf{(b)}  For a given $t \in \E_T$, we have
\begin{equation}
	t(g) = \left|\det g\right|^s P(g)\quad \text{for all}\quad g \in L_n\:,
\end{equation}
where $P(g)$ is a homogeneous $T$ valued polynomial in the components of
$g$ (with respect to the canonical basis of $\R^n$), and $s$ is some real
number fixed by  weight  of the tensor density representation
of $GL(n)$.

\textbf{(c)} For a given $t \in \tilde{\E}_T$, we have
\begin{equation}
	t(g) = \left|\det g\right|^s P(g,\eps(g))\quad
		\text{for all}\quad g \in L_n\:,
\end{equation}
where $P(g,\eps(g))$ is a homogeneous $T$ valued polynomial in the
components of $g$ and the components%
	\footnote{The homogeneous degree of $P(g,\eps(g))$ counts the
	components of $g$ with degree $2$ and the components of $\eps(g)$ with
	degree $n$.} %
of $\eps(g) := \sqrt{\det g}\, \epsilon$ (with respect to the
canonical basis of $\R^n$ in both cases), and $s$ is some real number
fixed by the weight of the tensor density representation of $GL^+(n)$.
\end{lem}
\begin{rem}\label{rem:eqtd-hom}
Since, in view of this Lemma, an equivariant tensor density $t(g)$ is a
homogeneous function, say of degree $k$, of $g$ up to a power of
$\left|\deg g\right|$, it could always be rewritten as
\begin{equation}
	t(g) = \left|\det g \right|^{\frac{k}{n}}
		t( \left|\det g\right|^{-\frac{1}{n}} g) .
\end{equation}
This observation will be later useful in the proof of
Theorem~\ref{thm:main}.
\end{rem}
\begin{proof}
We deal with the $GL(n)$-equivariant case, the $GL^+(n)$-equivariant
case being completely analogous. The action of $GL(n)$ on both $L_n$ and
$T$ is linear, so we denote it as $u\cdot x$, for $u\in GL(n)$ and $x$
in either $L_n$ or $T$.

The first crucial observation, as $L_n$ consists of a single orbit of
$GL(n)$, is that equivariance allows us to fully fix $t\colon L_n \to T$
provided that we know its value on $\eta \in L_n$, by the formula
\begin{equation}\label{eq:tens-isom}
	t(g) = t(u_g\cdot \eta) = u_g\cdot t(\eta) ,
\end{equation} 
for any $g\in L_n$ and $u_g \in GL(n)$ such that $g = u_g \cdot \eta$.
The second crucial observation is that, to make sure that the values of
$t$ are assigned consistently, $t(\eta)$ must be invariant under the
isotropy subgroup of $\eta$, namely $O(1,n-1)$. In other words,
$t(\eta)$ must belong to $\I_{T}$, with respect to the induced
representation of $O(1,n-1)$ on $T$. The formula~\eqref{eq:tens-isom}
clearly defines mutually inverse maps $\E_T \to \I_T$ and $\I_T \to
\E_T$, thus establishing the isomorphism $\E_T \cong \I_T$ claimed in
part (a).

Let us now prove part (b). Fix an (equivariant) embedding $\alpha\colon
T\to M^p_n$. Since $t(\eta)$ is an element of $\I_{T}$, from the
characterization of isotropic tensors in Proposition~\ref{prp:iso-tens},
it must be of the form
\begin{equation}
	t(\eta) = \alpha^{-1}
		\left(\sum_{\sigma\in S_p} c^\sigma \eta_{\sigma} \right) ,
\end{equation}
where $c^\sigma$ are some scalar coefficients. Then for any $g\in L_n$
and a corresponding  $u_g\in GL(n)$ such that $g = u_g \cdot \eta$,
\begin{align}
\notag
	t(g) = u_g \cdot t(\eta)
		&= \alpha^{-1} \left( \left| \det u_g \right|^r
				\sum_{\sigma\in S_p} c^\sigma (u_g\cdot \eta_\sigma) \right) \\
		&= \left| \det g \right|^{r/2} \alpha^{-1} \left(
				\sum_{\sigma\in S_p} c^\sigma g_\sigma \right) , \label{yestag}
\end{align}
where $r$ is the density weight of the representation $T$ and where we
have used the notation
\begin{equation}
	g_\sigma = g_{\sigma(i_1)\sigma(i_2)} \cdots g_{\sigma(i_{p-1})\sigma(i_p)}
\end{equation}
for the corresponding monomial on $L_n$ in terms of the components of
$g$ with respect to the canonical basis on $\R^n$. Clearly, the above
formula can be rewritten as $t(g) = \left| \det g \right|^s P(g)$, with
$s := r/2$. We observe that, from~\eqref{yestag}, that  $P$ is an
homogeneous polynomial (of degree $p/2$) in the components of the
metric, completing the proof of part (b).

The proof of (c) is strictly analogous, taking into account the
identity $u_g\cdot \epsilon = \eps(g)$, for any $u_g \in GL^+(n)$ such
that, $u_g\cdot \eta = g$.
\end{proof}

\section{Characterization of Finite Renormalizations of Wick Polynomials}
\label{sec:wick}

We generalize the discussion of local covariant fields
from~\cite{hw-wick}, where only metric dependence was allowed, to a more
general context where other \emph{background fields} are allowed in
addition to the metric  $\g$ on a spacetime $M$. In order to simplify
the presentation, we will restrict the extra background fields to two
scalar functions $m^2$ and $\xi$, which appear in the description of a
scalar quantum field.

Generally speaking, background fields are described by sections $\h$ of
suitable bundles $HM \to M$ over the manifolds $M$ we consider.
Covariance requires us to deal with all such bundles \emph{simultaneously}
and \emph{coherently}. In other words  we deal with  an assignment of a
bundle $HM \to M$ to \emph{every} manifold $M$ and require that  any
embedding $\chi \colon M \to M'$ must give rise to a corresponding
well-defined pullback map $\chi^*\colon \Secs(HM') \to \Secs(HM)$.  This
picture can be phrased properly with the language of category theory by
means of the notion of \emph{natural bundle}. Before giving the
definition, we would like to require a bit more geometric structure from
the bundles of background fields that interest us. A bundle $F\to M$ is
\textbf{dimensionful} if there it has an action $\R^+ \times F \to F$ of
the multiplicative group $\R^+$ of positive real numbers, called
\emph{(physical) scaling}, which acts by a diffeomorphisms that fix each
fiber of $F\to M$. Any vector bundle is automatically dimensionful, by
virtue of having a well-defined multiplication by scalars on its fibers,
although we will not always use this particular scaling action. To avoid
confusion, we should mention that a \emph{dimensionless} bundle would be
a special kind of dimensionful bundle, where scaling transformations act
trivially.

A \textbf{natural (dimensionful) bundle} is a functor $H\colon \Man\to
\Bndl$ from the category of smooth manifolds (where objects are
connected, have fixed dimension $n$ and morphisms are embeddings, which
are necessarily local diffeomorphism) to the category of dimensionful
smooth bundles (where morphisms are bundle maps, i.e.,\ fiber
preserving, equivariant with respect to scaling), such that a morphism
$\chi \colon M\to M'$ induces a morphism $H\chi \colon HM \to HM'$ that
is itself a local diffeomorphism. The required pullback $\chi^*\colon
\Secs(HM') \to \Secs(HM)$ is then implicitly defined by $\h' \circ \chi
=  H\chi \circ (\chi^* \h')$, when $\h' \in \Secs(HM')$. The
equivariance of the morphism $H\chi$ ensures that scaling commutes with
the pullback, $\chi^*(\h'_\lambda) = (\chi^*\h')_\lambda$, for $\lambda \in
\R^+$.

One elementary example of a natural bundle is the functor $M \mapsto
\R\times M$, the trivial scalar bundle, whose sections we call scalar
fields, with scaling being simple multiplication. Another relevant
example is $M \mapsto \mathring{S}^2T^*M$, the bundle of Lorentzian
metrics; we will denote a section of $\mathring{S}^2T^*M \to M$ by $\g$.
Other examples are are $M\mapsto T^*M$ and $M\mapsto \Lambda^2 M$, the
cotangent bundle and the bundle of $2$-forms, whose sections could be
interpreted as background electromagnetic fields, in the vector
potential or field strength forms.  All of these bundles are
dimensionful, by virtue of being vector bundles, with the exception of
$\mathring{S}^2T^*M$, which inherits a scaling action from being
considered as a scaling invariant sub-bundle of the vector bundle
$S^2T^*M$.

\begin{rem}\label{remarkH}
In the rest of the paper, focussing  on the theory of a \emph{real
quantum scalar field}, $\varphi$,  we make a more precise choice of the
natural functor $H$. We suppose that the manifolds of the category
$\Man$ are connected, $n$-dimensional (for a fixed $n\geq 2$), and the
functor $H$ assigns $M\mapsto HM = \mathring{S}^2T^*M\times \R\times
\R$, with a morphism $\chi \colon M\to M'$ inducing the standard tensor
push-forward $H\chi = \chi_* \colon HM \to HM'$. Then, the sections
$M\to HM$ are triples $\h = (\g, m^2, \xi)$, scaling as $(\g, m^2, \xi)
\mapsto (\lambda^{-2}\g, \lambda^2 m^2, \xi)$, always consisting of:

 (a) a \emph{Lorentzian metric}, $\g$, making $(M,\g)$ a \emph{(smooth)
 Lorentzian spacetime} of fixed dimension $n\geq 2$,

 (b) the pair of real scalar fields $m^2$  and  $\xi$ over $M$, with the
 respective  physical meaning of the \emph{squared mass} of the scalar
 field and a factor describing the \emph{coupling with the scalar
 curvature}. 
\end{rem}

We stress that, exactly as in \cite{hw-wick}, we assume that the {\em
parameters} $m^2$ and $\xi$ are actually \emph{functions} on $M$.
Quantum field theory in curved spacetime is well-defined for both
constant or variable $m^2$ and $\xi$. There is of course no obstacle in
restricting them to constant functions, as we note in
Remark~\ref{rem:m2xi-param}. Moreover, as in \cite{hw-wick}, $m^2$ and
$\xi$ are allowed to have any real value.

\begin{defn}
Let us fix the natural bundle  $H\colon \Man \to \Bndl$ as in
Remark~\ref{remarkH}. A \textbf{background field} is a section $\h\colon
M \to HM$ and we call the pair $(M,\h)$ a \textbf{background geometry},
provided $\h = (\g, m^2, \xi)$ is such that $(M,\g)$ is a
time-orientable globally hyperbolic spacetime.  Furthermore we define
the following categories.

\textbf{(a)} $\BkgG$ is the \textbf{category of background geometries},
having time-oriented background geometries as objects and morphisms
given by smooth embeddings $\chi\colon M\to M'$ that preserve the
background fields, $\chi^* \h = \h'$ on $M'$, the time orientation, and
causality, meaning that every causal curve between $\chi(p)$ and
$\chi(q)$ in $M'$ is the $\chi$-image of a causal curve between $p$ and
$q$ in $M$. 

\textbf{(b)}  $\BkgG^+$ is  the \textbf{category of oriented background
geometries} having  oriented and time-oriented background geometries as
objects and morphisms as in $\BkgG$, but also required to preserve the
spacetime orientation.

Since the natural bundle $H$ is dimensionful, scaling transformations
also act on these categories by $(M,\h) \mapsto (M,\h_\lambda)$, for any
$\lambda \in \R^+$, which by equivariance of the pullback of background
fields act as functors, $\BkgG \to \BkgG$ and $\BkgG^+ \to \BkgG^+$
respectively.
\end{defn}

To describe the algebras of observables on background geometries, we
need the notion of a \emph{net of algebras} (or {\em pre-cosheaf of
algebras}). 

\begin{defn}\label{def-loc-cov-field}
A \textbf{net of algebras (of observables)} is an assignment of a
complex unital $*$-algebra $\W(M,\h)$ for every background geometry
$(M,\h)$ in $\BkgG$ together with an assignment of an injective unital
$*$-algebra homomorphism $\iota_\chi \colon \W(M,\h) \to \W(M',\h')$ for
every morphism in $\BkgG$, respecting compositions. In other words $\W\colon\BkgG \to \Alg$ is a functor
from the category of background geometries into the category of {(complex)
unital $*$-algebras} whose morphisms are injective unital $*$-algebra
homomorphisms. Further, we require
that $\W$ respects \emph{scaling} and the \emph{time slice axiom}.
\begin{itemize}
\item[\textbf{(i)}]
	Scaling transformations $(M,\h) \mapsto (M,\h_\lambda)$ result in
	$*$-algebra isomorphisms $\sigma_\lambda\colon \W(M,\h) \to
	\W(M,\h_\lambda)$. Scaling transformations act as natural isomorphisms
	$\sigma_\lambda \colon \W \to \W_\lambda$ between the $*$-algebra
	valued functors $\W$ and $\W_\lambda$, the latter defined by
	$\W_\lambda(M,\h) = \W(M,\h_\lambda)$.
\item[\textbf{(ii)}]
	Given a morphism $\chi\colon (M',\h') \to (M,\h)$ of background
	geometries, if the image $\chi(M') \sse M$ contains a Cauchy surface
	for $(M,\g)$, then the induced $*$-homomorphism $\iota_\chi \colon
	\W(M',\h') \to \W(M,\h)$ is a $*$-isomorphism.
\end{itemize}
We refer to a functor $\W \colon \BkgG^+ \to \Alg$ with the same
properties as a net of algebras as well.
\end{defn}

The algebras of observables $\W(M,\h)$ are intended to be the algebras
of Wick products, whose construction and the fact that they satisfy all
the desired properties are discussed in detail in~\cite[Sec.~2]{hw-wick},
cf.\ their Lemma~4.2 in particular, where they define the scaling
isomorphism $\sigma_\lambda$. However, we shall not touch
upon these details and rely only on the properties of the $\W$ functor
as axiomatized above and also below in the definition of Wick powers.

Having defined a net of algebras, respecting local covariance and
scaling, we are in a position to state our definition of a \emph{locally
covariant (almost) homogeneous} quantum field which somewhat extends and
generalizes~\cite[Def.~3.2]{hw-wick}.
\begin{defn}\label{def1}
A \textbf{locally covariant scalar quantum field} $\Phi$ is an
assignment of an algebra-valued distribution%
	\footnote{For every Hadamard quasifree  state $\omega$ over $W(M,\g)$
	the map $\D(M) \ni f \mapsto \omega(\Phi(f))$ is a distribution in the
	proper sense. A weaker requirement allowing to smear fields  with
	distributions of a suitable wavefront set can be given exploiting the
	so called H\"ormander pseudotopology \cite{hw-wick}, but it is
	irrelevant for this work.} %
$\Phi_{(M,\h)} \colon \D(M) \to \W(M,\h)$ to each background geometry
$(M,\h)$ that satisfies the following identity for each morphism
$\chi\colon (M',\h'=\chi^*\h) \to (M,\h)$:
\begin{equation}\label{loc-cov}
	\iota_\chi (\Phi_{(M,\chi^*\h)}(f))
		= \Phi_{(M,\h)}(\chi_*f)  \:, \quad \text{for any $f \in \D(M')$.}
\end{equation}
In other words, $\Phi$ is a natural transformation $\Phi\colon \D \to
\W$ between the functor of test functions%
	\footnote{Note that $\D\colon \Man \to \LCV$ is a (covariant) functor
	from manifolds to locally convex topological vector spaces. It assigns
	the space of complex valued test functions $\D(M)$ to a manifold $M$
	and maps a morphism $\chi\colon M\to M'$ to the induced
	\emph{extension by zero} map $\chi_*\colon \D(M) \to \D(M')$.} %
and the net of algebras of observables (thus, another name for $\Phi$
could be a \emph{natural} scalar quantum field).
\end{defn}

The reason that a quantum field $\Phi$, even on a fixed spacetime
$(M,\g)$, is associated with an algebra-valued distributions is the
usual heuristic according to which pointlike fields $\Phi(x)$ are too
singular to be evaluated directly, while its smearing with a test
function $f\in \D(M)$,
\begin{equation}
	\Phi(f) = \int_M \Phi(x) f(x) \, dg(x) ,
\end{equation}
where $dg(x)$ is the volume form induced by the metric $\g$, is a
legitimate observable. Where appropriate, we will use the distributional
notation $\Phi(x)$ as well.

Note the most trivial example of a locally covariant scalar quantum
field, which we may call the \emph{unit $c$-number field} $\mathbf{1}$,
defined by the formula
\begin{equation}\label{eq:c-num}
	\mathbf{1}_{(M,\h)} (f) := 1 \int_M f(x)\, dg(x) ,
\end{equation}
where $1\in \W(M,\h)$ is the algebra unit. If $C[\h]$ is any function
that maps a background geometry on $M$ to a distribution on $M$ and
satisfies the identity $\chi^*C[\h] = C[\chi^*\h]$ for any morphism
$\chi \colon (M',\chi^*\h) \to (M,\h)$, then the product $C\mathbf{1}$
functor defined by the formula
\begin{equation}
	C\mathbf{1}_{(M,\h)} (f) = 1 \int_M C[\h](x) f(x) \, dg(x) ,
\end{equation}
is also a locally covariant scalar quantum field that we refer to as a
\emph{$c$-number} field.

\begin{rem}\label{rem-loc-cov} \strut
\begin{itemize}
\item[\textbf{(1)}]  In \cite{hw-wick}, $\h$ is nothing but the
Lorentzian metric of the  spacetime and the parameters $m^2$ and $\xi$
appearing in the definition of  the quantum fields generated by KG
fields are considered external parameters. Here instead  we explicitly
include them in $\h$.  It is very easy to prove that the concrete
locally covariant quantum fields appearing in \cite{hw-wick} (scalar KG
field and associated  Wick polynomials, time-ordered Wick polynomials
and their derivatives) satisfy our somewhat more general definition
locally covariant quantum fields.
\item[\textbf{(2)}] Definition \ref{def-loc-cov-field} includes two
distinct though related  notions: \emph{locality} and \emph{covariance},
both illustrated by the condition~\eqref{loc-cov}. Locality corresponds
to the case where $\chi$ describes the inclusion $\chi\colon  M \sso
M'$, while covariance corresponds to an arbitrary allowed $\chi$.
\end{itemize}
\end{rem}

In~\cite[Sec.~2]{hw-wick}, the algebras $\W(M,\h)$ are constructed
explicitly for the case of the quantization of a Klein--Gordon scalar
field $\varphi$, with mass $m^2$ and coupled with the scalar curvature
through the constant $\xi$, whose equation of motion is
\begin{equation}\label{eq:KG}
	\square_\g \varphi - m^2 \varphi - \xi R_\g \varphi = 0  .
\end{equation}
In that context, the basic example of a locally covariant scalar field
is the Klein--Gordon field $\varphi$ itself. Since we are not dealing
with such explicit constructions, we simply encapsulate the needed
properties of $\varphi$ in the definition below.
\begin{defn}\label{def:lin-field}
Given a net of algebras $\W$ on $\BkgG$ or $\BkgG^+$, a \textbf{linear
quantum scalar field} $\varphi$ is a locally covariant quantum scalar
field that satisfies the following \textbf{kinematic completeness}
property: For any $(M,\h)$, an element $a\in \W(M,\h)$ satisfies
$[a,\varphi_{(M,\h)}(f)] = 0$ for every $f\in \D(M)$ iff $a = \alpha\, 1$,
with $\alpha\in \mathbb{C}$ and $1$ the unit element of the algebra.
\end{defn}
The above definition is rather minimal and, certainly, the linear KG
field as defined in~\cite[Sec.~2]{hw-wick} satisfies several additional
properties, but we have omitted most of them. Below, we will define
\emph{Wick powers} $\varphi^k$, which will include the linear field
$\varphi^1 := \varphi$ as a special case. Within the definition to
follow, we will require further properties to hold for $\varphi^k$ for
each $k$, including $k=1$, thus imposing further axioms also on the
linear field and bringing our axiomatization closer to that
of~\cite{hw-wick}. Thus, to avoid some repetition, we find it more
economical to state these axioms in a way that is uniform in $k$. As it
is the main goal of this work to remove it, the \emph{analyticity} axiom
will not appear below and the \emph{continuity} axiom will be suitably
modified to compensate. On the other hand, we will not make use of the
\emph{on-shell} condition, $\varphi((\square_\g - m^2 - \xi R_\g)f) =
0$, though the equations of motion appear implicitly through the
\emph{time slice axiom} in Definition~\ref{def-loc-cov-field}. We have
excluded it because it plays no role in our analysis below, since it is
restricted to simple Wick powers. However, the on-shell condition will
have to be taken into account when expanding the analysis to more
general Wick and time-ordered products that involve derivatives of
$\varphi$.

Before giving a precise axiomatic definition of Wick powers, we need to
address the technical question of how physical scalings and continuous
variations of the background geometry can be made to act on locally
covariant scalar quantum fields.

First, we address scalings. Given a locally covariant scalar quantum
field $\Phi$, we can define a new rescaled locally covariant scalar
quantum field $S_\lambda \Phi$ by the formula
\begin{equation}\label{eq:scal-phi-def}
	(S_\lambda\Phi)_{(M,\h)}(f)
		= \sigma_\lambda^{-1} \left(\Phi_{(M,\h_\lambda)}(\lambda^n f)\right) ,
\end{equation}
where $\lambda \in \R^+$, $n=\dim M$, and $\sigma_\lambda$ is the
isomorphism realizing the action of scalings on the net of algebras of
observables. The extra factor of $\lambda^n$ compensates for the fact
that integration against the test function $f$ is done with respect to
the metric volume form $dg$, which scales as $dg \mapsto \lambda^{-n}
dg$ when $\g \mapsto \lambda^{-2} \g$. Comparing this formula
with~\cite[Eq.~(48)]{hw-wick}, note that the direction of our isomorphism
$\sigma_\lambda$ is opposite. More formally, recall that $\sigma_\lambda
\colon \W \to \W_\lambda$ is a natural isomorphism between two functors
(defined on $\BkgG$ or $\BkgG^+$), while $\Phi \colon \D \to \W$ is
another natural transformation. If we similarly define the natural
linear isomorphism $\mu_\lambda \colon \D \to \D$ given by $f \mapsto
\lambda^n f$, the rescaled quantum field is given by a composition of
these natural transformations, $S_\lambda \Phi = \sigma_\lambda^{-1}
\circ \Phi \circ \mu_\lambda$. As such, we have defined a representation
$S_\lambda$ of the multiplicative group $\R^+$ on the vector space%
	\footnote{We make the tacit and harmless assumption that we are
	working with \emph{small categories}, whose objects and morphisms
	constitute sets.} %
of natural transformations $\D \to \W$, which opens the door to using
the generic notion of homogeneous and almost homogeneous elements from
Definition~\ref{def:almost-hom} to quantum fields.

Next, in order to address the continuity hypothesis, we must specify how
to identify the algebras of observables defined on different background
geometries, as was done in~\cite[Sec.~4.2]{hw-wick}. Let $(M,\h_1)$ and
$(M,\h_2)$ be two background geometries that differ only inside a
compact set $O\sso M$, recalling that $\h_i = (\g_i,m^2_i,\xi_i)$. It is
a simple fact of Lorentzian geometry that the spacetimes
$(M_\pm,\g_\pm)$, with $M_- = M \setminus J^+(O)$, $M_+ = M \setminus
J^-(O)$ and $\h_\pm = \h_1|_{M_\pm} = \h_2|_{M_\pm}$, are globally
hyperbolic spacetimes in their own right. Moreover, each of the
$(M_\pm,\h_\pm)$ contains a Cauchy surface in common with both
$(M,\h_1)$ and $(M,\h_2)$. Thus, denoting by $\chi^i_\pm \colon (M_\pm,
\h_\pm) \to (M,\h_i)$ the inclusion morphisms, the time slice axiom of
the net of the algebras observables gives us isomorphisms
$\iota_{\chi^i_\pm} \colon \W(M_\pm,\h_\pm) \to \W(M,\h_i)$. These
isomorphisms allow us to identify the algebras of observables of
background geometries $(M,\h_1)$ and $(M,\h_2)$ that differ only inside
a compact $O\sso M$ in two ways, $\tau_\mathrm{ret} , \tau_\mathrm{adv}
\colon \W(M,\h_1) \to \W(M,\h_2)$, where these $*$-algebra isomorphisms
are defined by $\tau_\mathrm{ret} = \iota_{\chi^-_2} \circ
\iota_{\chi^-_1}^{-1}$ and $\tau_\mathrm{adv} = \iota_{\chi^+_2} \circ
\iota_{\chi^+_1}^{-1}$.  Below, we will only make use of
$\tau_\mathrm{ret}$, though choosing $\tau_\mathrm{adv}$ would not have
lead to equivalent results.

Finally, we introduce an axiomatic definition of \emph{Wick powers},
which will be our main objects of interest. We deviate somewhat in our
axiomatization from the analogous one in~\cite[Sec.~4]{hw-wick} for
reasons expanded on below.

\begin{defn}[Wick powers]\label{def:wick}
Given a net of algebras $\W$ (Definition~\ref{def-loc-cov-field}) on the
category of background geometries $\BkgG$ (or $\BkgG^+$) and a
corresponding locally covariant linear scalar quantum field $\varphi$
(Definition~\ref{def:lin-field}), we define \textbf{Wick powers}
$\{\varphi^k\}$ of $\varphi$, for $k=0,1,2,\cdots$, by the following
axioms:
\begin{itemize}
\item[\textbf{(i)}] \textbf{Locality and Covariance.}
Each Wick power $\varphi^k$ is a locally covariant scalar quantum field
(Definition~\ref{def1}), which for low powers agree with $\varphi^0 =
\mathbf{1}$, the unit $c$-number field, and $\varphi^1 = \varphi$, the
linear field.
\item[\textbf{(ii)}]  \textbf{Scaling.}
Each Wick power $\varphi^k$ is almost homogeneous of degree $k(n-2)/2$
(Definition~\ref{def:almost-hom}) with respect to the action of physical
scalings $S_\lambda$ (Eq.~\eqref{eq:scal-phi-def}) on locally
covariant fields. That is, there exists an integer $l\ge 0$ and locally
covariant fields $\psi_j$ such that $S_\lambda \varphi^k =
\lambda^{k\frac{(n-2)}{2}} \varphi^k + \lambda^{k\frac{(n-2)}{2}}
\sum_{j=1}^l (\log^j \lambda) \psi_j$, where each $\psi_j$ is itself
almost homogeneous of degree $k(n-2)/2$ and order $l-j$ (as in
Eq.~\eqref{eq:almost-hom-def}). In dimension $n=4$, $(n-2)/2=1$
gives the usual scaling of a scalar field.
\item[\textbf{(iii)}]  \textbf{Algebraic.}
Each Wick power $\varphi^k$ also satisfies the following properties:
	\begin{itemize}
	\item[]  \textbf{Hermiticity:}
	\begin{equation}
		\varphi_{(M,\h)}^k(f)^* = \varphi_{(M,\h)}^k(f^*),
	\end{equation}
	where on the left $*$ denotes the corresponding operation in the
	$*$-algebra $\W$, while on the right $*$ denotes simple complex
	conjugation.
	\item[]  \textbf{Commutator expansion:}
		\begin{equation}
			[\varphi_{(M,\h)}^k(x),\varphi_{(M,\h)}(y)]
				= i k \varphi^{k-1}(x) \Delta_{(M,\h)}(x,y),
		\end{equation}
		where $\Delta_{(M,\h)}(x,y) = G^+_{(M,\h)}(x,y) - G^-_{(M,\h)}(x,y)$
		is the difference between the retarded and advanced Green functions
		for the KG equation~\eqref{eq:KG}.
	\end{itemize}
\item[\textbf{(iv)}]  \textbf{Parametrized microlocal spectrum condition.}
	Consider any background geometry $(M,\h)$ and any smooth $m$-parameter
	family $(M,\H_s)$ of compactly supported variations thereof ($m\ge
	0$). That is, $\h = \H_0$, where $\H_s(x) = \H(s,x)$ with smooth $\H
	\colon \R^m \times M \to HM$, such that $\H_s$, for any $s\in \R^m$,
	differs from $\h$ only on a compact subset $O \sso M$. Further, let us
	implicitly identify each algebra of observables $\W(M,\H_s)$ with
	$\W(M,\h)$ using the isomorphism $\tau^s_\mathrm{ret}\colon \W(M,\H_s)
	\to \W(M,\h)$ discussed earlier.

	Then, for any quasi-free Hadamard state $\omega$ on $\W(M,\h)$ with
	respect to $\varphi$ and each Wick power $\varphi^k$, the expectation
	value $\omega(\varphi^k(x))$ is a distribution on $\R^m\times M$ with
	empty wavefront set, which hence can be represented by a smooth
	function.
\end{itemize}
\end{defn}

\begin{rem}\label{rem:wick}
The above axioms for Wick powers differ in some aspects from those given
in~\cite[Sec.~4]{hw-wick}.
\begin{itemize}
\item[(a)]
Our \emph{scaling} condition, which uses Definition~\ref{def:almost-hom}, is
slight\-ly weaker than Definition~4.2 of~\cite{hw-wick}, but it will be
sufficient for our purposes. The difference is in the notion of
\emph{order} of the logarithmic terms. The `order' in
Definition~\ref{def:almost-hom} refers only to the scaling properties.
On the other hand, the `order' of a quantum field used
in~\cite[Def.~4.2]{hw-wick} refers to the number of iterated commutations
with $\varphi$ needed to annihilate that field. An inductive argument
(in $k$) shows that if a Wick power $\varphi^k$ satisfied Definition~4.2
of~\cite{hw-wick}, then the same Wick power satisfies also our
Definition~\ref{def:almost-hom}(b).
\item[(b)]
The technical \emph{continuity} and \emph{analyticity} conditions were
replaced by a strengthened version of the \emph{microlocal scaling
condition}. Using the notation from the last point of
Definition~\ref{def:wick}, the continuity condition in~\cite{hw-wick}
required $\tau^s_\mathrm{ret} \circ \varphi^k_{(M,\H_s)}(f) \in \W(M,\h)$
to be continuous in $s$ for any test function $f\in \D(M)$ and any
$1$-parameter family $\H(s,x)$, with the topology on the algebra
$\W(M,\h)$ left implicit. The analyticity condition required the
expectation values $\omega^s(\varphi^k_{(M,\H_s)}(x))$ to be analytic in
$(s,x)$ whenever the $1$-parameter family $\H(s,x)$ is analytic (of
course relaxing the requirement that the variations of the background
geometry must vanish outside a compact set) and is accompanied by a
$1$-parameter family of states $\omega^s$ analytic in $s$, with
$\omega^s$ quasi-free and Hadamard on $\W(M,\h)$ and both the linear
structure and topology on the appropriate space of states left implicit.
\item[(c)]
Removing the analyticity condition is our main goal, so that change
should not be surprising. In the proof of Theorem~\ref{thm:main} we will
appeal to the Peetre--Slov\'ak theorem, instead of analyticity as was
done in~\cite{hw-wick}, to conclude that ambiguities in the Wick powers
are characterized by differential operators. Thus, the remaining
technical continuity hypothesis is needed only in as much as it helps
meet the \emph{weak regularity} hypothesis needed for the application of the
Peetre--Slov\'ak theorem, which is the only one not already covered by
\emph{locality}. However, using the notation from the last point of
Definition~\ref{def:wick}, the continuity of $\varphi_{(M,\H_s)}^k(f)
\in \W(M,\h)$ as a function of $s$, even with an opportune choice of
topology on $\W(M,\h)$, is not sufficient to assure the needed weak
regularity property of $\varphi^k$. What would be needed instead is a
technical infinite-dimensional \emph{smoothness} condition. Instead of
going down that road, we simply strengthen the original \emph{microlocal
spectrum condition} to cover parametrized families of, instead of just
individual, background geometries. This new \emph{parametrized
microlocal condition} hypothesis then essentially directly yields the
needed weak regularity property of $\varphi^k$.
\end{itemize}
\end{rem}

Along with the explicit construction of the algebras $\W(M,\h)$ in
\cite{hw-wick}, Hollands and Wald also shows the existence of a family
of locally covariant scalar quantum fields that satisfy their version of
the axioms of Wick powers, whose difference from ours are discussed
above. However, that given construction is not the only way to satisfy
these axioms. The lack of uniqueness is physically interpreted as the
existence of some remaining degrees of freedom in the renormalization
procedure of Wick powers. The key result of~\cite{hw-wick} on finite
renormalizations of Wick powers is stated in \cite[Thm.~5.1]{hw-wick}. 
Any pair of  families of  Wick powers%
	\footnote{The differences between the Hollands and Wald definition of
	Wick powers from ours is detailed in Remark~\ref{rem:wick}.} %
$\{\tilde{\varphi}^k\}$ and $\{\varphi^k\}$ ($k\in \mathbb{N}$) of the
same Klein--Gordon field $\varphi = \varphi^1= \tilde{\varphi}^1$
satisfies the following relation, in our notation, for every fixed
background geometry $(M,\h)$:
\begin{equation}\label{eq:wick-diff}
	\tilde{\varphi}^k_{(M,\h)}(x)
	= \varphi^k_{(M,\h)}(x)
		+ \sum_{i=0}^{k-2} \binom{k}{i} C_{k-i}[\h](x) \varphi_{(M,\h)}^i(x) \:,
\end{equation}
Above, the scalar coefficients $C_k[\h]$ are some scalar differential
operators that are tensorially constructed out of the metric, the
curvature tensor and its derivatives. These operators depend
polynomially on the curvature tensor, its derivatives and on $m^2$, with
coefficients that depend \emph{analytically} on $\xi$. Moreover, the
$C_k[\h]$ scale as  $C_k \mapsto \lambda^{k\frac{(n-2)}{2}} C_k$ when
their arguments are rescaled as $\xi\mapsto \xi$, $m^2\mapsto \lambda^2
m^2$, $\g^{ab} \mapsto \lambda^2 \g^{ab}$ and $R_{abcd} \mapsto
\lambda^{-2} R_{abcd}$, with the same scaling weight for its
derivatives. In dimension $n=4$, the coefficients scale as $C_k \mapsto
\lambda^k C_k$, which is the expression that appears in~\cite{hw-wick}.

As mentioned in the Introduction, the Analyticity requirement that
distinguishes the Hollands and Wald definition of Wick powers from ours
(Remark~\ref{rem:wick}) is somewhat unnatural and technically difficult
to handle, as was stressed in the Introduction. We would like to
demonstrate that Analyticity is not necessary to prove a result that is
essentially similar to the statement of \cite[Thm.~5.1]{hw-wick}
mentioned above. On the other hand, the Continuity requirement of
Hollands and Wald cannot be completely dispensed with, since some
version of it must survive to feed into the weak regularity hypothesis
needed by the Peetre--Slov\'ak theorem. However, it seems difficult to
find a simple modification of Continuity that would do the job. Instead,
we find it more useful to strengthen the usual Microlocal Spectrum
Condition to its Parametrized version (Definition~\ref{def:wick}(iv)).
We believe that this strengthened version is rather natural,
encapsulating the stability of the properties of Wick powers under
variations of the parameters of the background geometry in precise
technical terms, without leaving the bounds of smooth differential
geometry.

In~\cite[Sec.~5.2]{hw-wick}, Hollands and Wald argue that the by now
standard locally covariant Hadamard parametrix prescription for defining
Wick powers satisfies all the requirements that we listed in
Definition~\ref{def:wick}, with our Parametrized Microlocal Spectral
condition replaced by the standard one and adding also their Continuous
and Analytic dependence requirements (Remark~\ref{rem:wick}). It should
be noted that they left the arguments supporting Continuous and Analytic
dependence implicit, not giving a complete proof, which appeared later
in~\cite{hw2}. Similarly, we believe that the locally covariant Hadamard
parametrix prescription satisfies also our Parametrized Microlocal
Spectral condition, but leave a detailed complete proof of this claim to
future investigations. Thus, by eliminating both the Continuity and
Analyticity requirements in favor of the Parametrized Microlocal
Spectrum Condition, we can achieve essentially the same result written
below into a more precise form:

\begin{thm}\label{thm:main}
Let $\{\tilde{\varphi}^k\}$ and $\{\varphi^k\}$  be two families ($k\in
\mathbb{N}$) of Wick powers, as in Definition~\ref{def:wick}, with
respect to a linear scalar quantum field $\varphi$
(Definition~\ref{def:lin-field}) in a net of algebras $\W$
(Definition~\ref{def-loc-cov-field}, defined on either the category
$\BkgG$ of non-oriented background geometries or $\BkgG^+$ of oriented
background geometries).

\textbf{(a)}  If $\varphi$ is defined with respect to the category
$\BkgG$, for every $(M,\h)$, the difference between $\tilde{\varphi}^k$
and $\varphi^k$ can be parametrized as in~\eqref{eq:wick-diff}, where
the coefficients
\begin{multline}\label{eq:Ck}
	C_k[\h](x) = C_k\left[\g^{ab}(x),   R_{abcd}(x), \ldots, \nabla_{e_1}\cdots
	\nabla_{e_{h}} R_{abcd}(x), \right. \\
\left. \xi(x), \ldots, \nabla_{e_1}\cdots
	\nabla_{e_{r}} \xi (x), m(x)^2, \ldots, \nabla_{e_1}\cdots
	\nabla_{e_{s}} m(x)^2\right] 
\end{multline}
are some scalar polynomials, tensorially formed from all of their
arguments, except $\xi(x)$, and where  $R_{abcd}(x)$ denotes the Riemann
tensor and $\nabla_a$ the Levi-Civita connection of $\g_{ab}$ at $x\in
M$.

\textbf{(b)}  If $\varphi$ is defined with respect to the category
$\BkgG^+$, for every $(M,\h)$, we have a variant of
\eqref{eq:wick-diff} with
\begin{multline}\label{eq:Ck+}
	C_k[\h](x) = C_k\left[\g^{ab}(x), \eps^{a_1\cdots a_n}(x),  R_{abcd}(x), \ldots, \nabla_{e_1}\cdots
	\nabla_{e_{h}} R_{abcd}(x), \right. \\
\left. \xi(x), \ldots, \nabla_{e_1}\cdots
	\nabla_{e_{r}} \xi (x), m(x)^2, \ldots, \nabla_{e_1}\cdots
	\nabla_{e_{s}} m(x)^2\right] 
\end{multline}
scalar polynomials, tensorially formed from all of their arguments,
except $\xi(x)$, and now including  the Levi-Civita tensor
$\eps^{a_1\cdots a_n}(x)$ of $\g_{ab}$ at $x\in M$.

In both cases (a) and (b), the coefficients of the polynomials are
smooth (instead of analytic) functions of $\xi(x)$ whose functional form
does not depend on $M$. 

Further, the $C_k$ scale as $C_k \mapsto \lambda^{k\frac{(n-2)}{2}}
C_k$ when their arguments are rescaled as follows: $\xi\mapsto \xi$,
$m^2\mapsto \lambda^2 m^2$, $g^{ab} \mapsto \lambda^2 g^{ab}$,
$\eps^{a_1\cdots a_n} \mapsto \lambda^n \eps^{a_1\cdots a_n}$,
$R_{abcd}(x) \mapsto \lambda^{-2} R_{abcd}(x)$ and the covariant
derivatives do not change this rescaling behaviour  as the coordinates
are  dimensionless. These rescaling properties fix the order of the
polynomial $C_k$.
\end{thm}

Obviously all terms $\nabla_{e_1}\cdots \nabla_{e_{s}} \xi (x)$ and
$\nabla_{e_1}\cdots \nabla_{e_{s}} m(x)^2$ with $s>0$ vanish if, at the
end of the computation, $m^2$ and $\xi$ are taken constant. Again, in
dimension $n=4$, the scaling dimension of the scalar field reduces to
the standard $(n-2)/2 = 1$.

The proof of our main Theorem~\ref{thm:main} will be mainly geometric.
However, we will need an intermediate analytical result, which we
encapsulate in the Lemma below, which is a more detailed version of the
first two paragraphs of the proof of~\cite[Thm.~5.1]{hw-wick}.
Logically, this analytical result follows from the Parametrized
Microlocal Spectrum property, from the Locality and Covariance
requirements (cf.~(2) in Remark~\ref{rem-loc-cov}) and from the Scaling
requirement. And, obviously, we make no use of either the Continuity or
Analyticity requirements from~\cite{hw-wick}, as we have replaced both
of those by our Parametrized Microlocal Spectrum property.
\begin{lem}\label{lem:locality}
For $\{\tilde{\varphi}^k\}$ and $\{\varphi^k\}$ as in
Theorem~\ref{thm:main} and every fixed $M$, the
identity~\eqref{eq:wick-diff} holds with some smooth functions
$C_k[\h]$, where the value $C_k[\h](x)$ depends only on the germ of $\h$
at $x\in M$. Moreover, these functions are locally covariant, so that
$\chi^* C_k[\h] = C_k[\chi^*\h]$ for any morphism $\chi$ in $\BkgG$
(resp.~$\BkgG^+$), also the functions $C_k$ are weakly regular in the
sense of Definition~\ref{def:reg} and $C_k[\h]$ scales almost
homogeneously of degree $k(n-2)/2$ under the physical scaling transformation
$\h = (\g, m^2, \xi) \mapsto (\lambda^{-2}\g, \lambda^{2} m^2, \xi)$.
\end{lem}

\begin{proof}[Proof of Lemma~\ref{lem:locality}]
The proof is inductive in $k$. The thesis holds for $k=1$ and $C_1=0$,
since $\varphi^1=\tilde{\varphi}^1 = \varphi$. Next suppose that
\eqref{eq:wick-diff} holds for some functions $C_i\colon \Secs(HM) \to
C^\oo(M)$, $i=1,2,\ldots, k-1$, that satisfy the desired properties.
Then $C_i[\h]1$ defines a locally covariant $c$-number field, where
$1\in \W(M,\h)$ is the identity of the given algebra. Define 
\begin{equation}\label{eq:wick-diff2}
	\Phi_{k,(M,\h)}(x) := \tilde{\varphi}^k_{(M,\h)}(x)
		- \left( \varphi^k_{(M,\h)}(x)
				+ \sum_{i=1}^{k-2} \binom{k}{i} C_{k-i}[\h](x) \varphi_{(M,\h)}^i(x)
			\right)\:.
\end{equation}
By construction, $\Phi_k$ is a locally covariant quantum field as in
Definition~\ref{def1} and also satisfies the \emph{Algebraic},
\emph{Scaling} and \emph{Microlocal} requirements of
Definition~\ref{def:wick}.  The
\emph{algebraic properties} in particular  require  that $\Phi_k$  is
Hermitian and, on any given spacetime $M$, it satisfies
$[\Phi_{k,(M,\h)}(x) , \varphi(y)]=0$ for all $x,y \in M$, which means
that it is a $c$-number field by the \emph{kinematic completeness}
property of $\varphi$ (Definition~\ref{def:lin-field}).  In other words,
$\Phi_{k,(M,\h)}=C_k[\h] 1$ where $C_k[\h]: C_0^\infty(M) \to \R$ is a
distribution.

Next, we appeal to the \emph{Parametrized Microlocal Spectrum}
condition. That is, considering $\h$ itself as a $0$-parameter family,
we can conclude that $C_k[\h](x) = \omega(C_k[\h](x) 1)$ is a smooth
function of $x$ for any Hadamard state $\omega$, since $\omega(1) = 1$
for any state and Hadamard states always exist. This establishes that we
have defined a map $C_k \colon \Secs(HM) \to C^\oo(M)$. If we introduce
an $m$-parameter family of compactly supported smooth deformation
$\H_s(x) = \H(s,x)$ of $\H_0 = \h$ then the same argument tells us that
$C_k[\H_s](x)$ is also jointly smooth in $(s,x)$. Thus, according to
Definition~\ref{def:reg}, the map $C_k$ is weakly regular.

The locality requirement  of Definition~\ref{def1} (see (2) in  Remark
\ref{rem-loc-cov}) entails that $\chi^* C_k[\h] =C_k[\chi^*\h]$ for any
inclusion $\chi\colon U\sso M$. In other words, fixing $x\in M$ and
taking the limit over decreasing neighborhoods $U$ of $x$, the value
$C_k[\h](x)$ depends only on the germ of $\h$ at $x$.

The validity of the \emph{Scaling} property for both $\varphi^k$ and
$\tilde{\varphi}^k$  imply that, by the formula~\eqref{eq:wick-diff2},
$\Phi_k$ is a linear combination of products of terms with almost
homogeneous degrees that add up to $k(n-2)/2$.  Thus, by
Lemma~\ref{lem:almost-hom-alg}, $\Phi_k$ itself has almost homogeneous
degree $k(n-2)/2$ and thus
\begin{equation}
	S_\lambda \Phi_k = \lambda^{k\frac{(n-2)}{2}} \Phi_k
		+ \lambda^{k\frac{(n-2)}{2}} \sum_i (\log^i \lambda) \Psi_i ,
\end{equation}
where $S_\lambda$ is the action of physical scalings on locally
covariant scalar quantum fields, with $\Psi_i$ some other locally
covariant quantum fields of almost homogeneous degree $k(n-2)/2$. See
Eq.~\eqref{eq:scal-phi-def}, and the discussion below it, for the
definition of $S_\lambda$ and in what sense locally covariant scalar
quantum fields form a vector space, so that
Definition~\ref{def:almost-hom} is applicable to them. Again, from the
kinematic completeness of $\varphi$, it follows that $\Psi_{i,(M,\h)} =
D_i[\h] 1$ are also all $c$-number fields. On the other hand, unwrapping
the definition of $S_\lambda$, we find that $S_\lambda (C_k[\h] 1) =
C_k[\h_\lambda] 1$, and similarly for the $D_i$. Hence, we find that
\begin{equation}
	C_k[\h_\lambda] = \lambda^{k\frac{(n-2)}{2}} C_k[\h]
		+ \lambda^{k\frac{(n-2)}{2}} \sum_i (\log^i \lambda) D_i[\h]
\end{equation}
is an almost homogeneous element of degree $k(n-2)/2$ of the space of maps
$\Secs(HM) \to C^\oo(M)$ under the action $D \mapsto
D_\lambda$, with $D_\lambda[\h] = D[\h_\lambda]$.
\end{proof}

In the proof of the main Theorem below, we systematically make use of
the geometric results summarized in Sect.~\ref{sec:geom}. In
particular, the \emph{Peetre--Slov\'ak theorem} discussed in
Sect.~\ref{sec:peetre} brings in the key simplification in our proof
in comparison with the arguments of~\cite{hw-wick}. This theorem is well
known in differential geometry but has not before been applied in this
context. It states that, under the conditions exhibited by
Lemma~\ref{lem:locality}, the $C_k$ must be some (possibly non-linear)
differential operators of locally bounded order applied to the
background fields $\g$, $m^2$ and $\xi$. It then remains only to call
upon the Scaling and Covariance properties to check that the $C_k$ may
only be of the form stipulated in Eq.~\eqref{eq:Ck}
or~\eqref{eq:Ck+}.

\begin{proof}[Proof of Theorem~\ref{thm:main}]
In this proof, we carefully separate the hypotheses of \emph{locality},
\emph{scaling} and \emph{covariance}. Locality allows us to conclude
that the functions $C_k$ are differential operators. Scaling restricts
their form and then covariance restricts their form even further, to the
desired result. Note that, unlike in~\cite{hw-wick} we do not make use
of Riemann normal coordinates. As a result, we invoke the
transformations properties of $C_k$ under two different kinds of scaling
transformations, which are mixed when normal coordinates are employed.

\emph{1. Locality and the Peetre--Slov\'ak theorem.} \quad
The first step is to combine the \emph{locality} of the coefficients
$C_k$ of Eq.~\eqref{eq:wick-diff} with the \emph{Peetre--Slov\'ak
theorem} (Proposition~\ref{prp:peetre}) to conclude that in fact these
coefficients are \emph{differential operators of locally bounded order}
(see Sect.~\ref{sec:peetre} for details). To verify the hypotheses of
Proposition~\ref{prp:peetre}, take the bundle $F \cong \R\times M\to M$,
so that its sections are just real valued functions $\Secs(F\to M) =
C^\oo(M)$. Finally, take the bundle $E \cong HM \to M$.
Lemma~\ref{lem:locality} shows that $C_k\colon \Secs(HM) \to C^\oo(M)$
such that $C_k$ is weakly regular and $C_k[\h](x)$ depends only on the
germ of $\h$ at $x \in M$. Consequently, the Peetre--Slov\'ak theorem
gives us the desired result: for every fixed $M \in \Man$, $C_k\colon
\Secs(HM) \to C^\oo(M)$ is a differential operator of locally bounded
order, as defined in Sect.~\ref{sec:peetre}.

Although we treat $m^2$ and $\xi$ as spacetime-dependent fields, this is
not crucial. They could be treated as constant parameters from the start
and the slight modification of the proof, needed only at this point, is
discussed in Remark~\ref{rem:m2xi-param}.

\emph{2. Almost homogeneity under physical scaling.} \quad
Consider a Lorentzian metric $\g_0$ on $M$, as well as a point $y\in M$
and an open neighborhood $U$ of $y$ with compact closure, with a
coordinate system $(x^a)$ centered at $y$. Since $C_k$ is a differential
operator of locally bounded order, for any such $\g_0$, $y$ and $U$
there exists%
	\footnote{In principle, the hypothesis of locally bounded order tells
	us that this is true for a sufficiently small neighborhood of $y$,
	with $r$ possibly increasing on larger neighborhoods. However, since
	such a neighborhood exists around any $y\in U$, a simple argument
	based on open covers and the compactness of $U$ shows that the
	order $r$ can be chosen uniformly over an arbitrary compact $U$.} %
an integer $r\ge 0$ such that $C_k$ is a differential
operator on $U$ of local order $r$ when acting on sections of $HM$ close
to $(\g_0,m^2=0,\xi=0)$, in a precise sense that we discuss next.
Naturally, the coordinates $x^a$ induce (scaling) adapted local
coordinates on the jet bundle $J^rHM$, which we write as $(x^a, g,
g_{ab}, g^{ab,A}, w^A, z^A)$, recalling that the coordinates
$(g,g_{ab})$ are functionally independent up to the identity $\left|\det
g_{ab}\right| = g$. The notation and the meaning of these coordinates
are discussed in Sect.~\ref{sec:coords}. The only difference is that
we now use two sets of coordinates, $w^A$ and $z^A$, for the jets of the
scalar fields, $m^2(x)$ and $\xi(x)$ respectively, instead of just one,
and that $w$ and $z$ have corresponding scaling degrees of $s=2$ and
$s=0$, as used in Sect.~\ref{sec:scal}. Then, by the bound $r$ on
the local order of $C_k$ at $y$, there exists a neighborhood $V_1^r\sse
J^rHM$ of $j^r_y(\g_0,m^2=0,\xi=0)$, projecting onto $U$, and a function
$F_k(x^a,g,g^{ab,A},w^A,z^A)$ defined on $V_1^r$ such that
\begin{equation}\label{eq:germ-jet-id}
	C_k[\h](x) = F_k(j^r\h(x)) ,
\end{equation}
for any section $\h\in \Secs(HM|_U \to U)$ such that $j^r\h(U) \sse
V^r_1$. Note that $V^r_1$ may be strictly smaller than $J^rH|_U$.
Without loss of generality, but possibly shrinking the domain of $F_k$,
we can choose it such that $V^r_1 \cong U\times W^r_1$, where the
projection on the $U$ factor is effected by the base coordinates $(x^a)$
and the projection onto $W^r_1$ is effected by the remaining fiber
coordinates. The main obstacle to increasing $V^r_1$ to all of $J^rHM$
is the possible need to increase the order $r$ on larger domains. At the
moment, from the Peetre--Slov\'ak theorem, we know only that the order $r$ of
$C_k$ is locally bounded, but may not have a finite global bound.  The
subscript $_1$ on $V^r_1$ will increase in the subsequent discussion as
we use the properties of $C_k$ to gradually enlarge the domain of
definition of the function $F_k$, while maintaining the
identity~\eqref{eq:germ-jet-id}, and thus the bound $r$ on the order of
$C_k$. In the final step of the proof we will in fact show that
differential order of $C_k$ is actually globally bounded. With that in
mind, it is then consistent, on a first reading of the proof, to assume
that $r$ is globally fixed and $V^r_1 = J^rHM$, so that the parts
dealing with enlarging $V^r$ could be skipped.

Similar to Eq.~\eqref{eq:scal1-vf}, the vector field implementing
infinitesimal physical scaling transformations on $V^r_1\sse J^rHM$ is
\begin{equation}
	e_1 = (2+2|A|) g^{ab,A} \del_{ab,A} + (2+2|A|) w^A \del^w_A
		+ 2|A| z^{A} \del^z_A .
\end{equation}
According to the last statement in Lemma~\ref{lem:locality} and an
immediate application of Lemma~\ref{lem:inf-almost-hom}, the coefficient
$C_k$ and hence the function $F_k$ scale almost homogeneously with
degree $k(n-2)/2$ with respect to the vector field $e_1$. Therefore, according
to Lemma~\ref{lem:almost-hom}, there exists an integer $l>0$ and
function $H_j$ on $V^r_1$, for $j=0,\ldots, l-1$, such that
\begin{equation}\label{eq:Fk-form}
	F_k = g^{-\frac{k(n-2)}{4n}} \sum_{j=0} \log^j (g^{-\frac{1}{2n}}) H_j ,
\end{equation}
where each $H_j$ is invariant under the action of $e_1$ and hence can be
written as
\begin{equation}\label{eq:H-def}
	H_j = H_j(x^a, g^{-\frac{1}{n}} g_{ab},
		g^{\frac{1}{2n} + \frac{1}{n}|A|} g^{ab,A},
		g^{\frac{1}{n} + \frac{1}{n}|A|} w^A, 
		g^{\frac{1}{n}|A|} z^A) .
\end{equation}
At this point, we may extend the domain $V^r_1$ to $V^r_2 \sse J^rHM$,
which is invariant under physical scaling. That is, we can write $V^r_2
\cong \R^+ \times W^r_2$, where the coordinate $g$ effects the
projection onto the $\R^+$ factor and the coordinates $(x^a,
g^{-\frac{1}{n}} g_{ab}, g^{\frac{1}{2n} + \frac{1}{n}|A|} g^{ab,A},
g^{\frac{1}{n} + \frac{1}{n}|A|} w^A, g^{\frac{1}{n}|A|} z^A)$ effect
the projection onto the $W^r_2$ factor, which includes at least the
point $(g^{-\frac{1}{n}} g_{ab}\circ \g_0(y), 0, 0, 0) $. The function
$F_k$ extends from $V^r_1$ to $V^r_2$ in a unique way as an almost
homogeneous function of degree $k(n-2)/2$.

Let us go into some of the details of the mentioned unique extension
procedure. So far, we could only presume that the
identity~\eqref{eq:germ-jet-id} that expresses the function $C_k[\h](x)$
in terms of the differential operator defined by the function $F_k$
holds only when the germ of $\h$ at $x\in M$ projects onto one of the
jets in the domain $V^r_1\sse J^rHM$ of $F_k$. We have defined the
extended domain $V^r_2$ to be the smallest domain invariant under
physical scaling and containing $V^r_1$. The function $F_k$, by using
formula~\eqref{eq:Fk-form}, has a unique almost homogeneous extension to
$V^r_2$ that scales almost homogeneously and agrees with the known
values of $F_k$ on $V^r_1$. Since any element of $V^r_2$ can be brought
back to $V^r_1$ by a physical scaling transformation and $C_k[\h]$
itself scales almost homogeneously, the identity~\eqref{eq:germ-jet-id}
must remain valid also for germs of $\h$ at $x$ that project to jets in
the extended domain $V^r_2$. Below, we use similar logic each time the
domain of the function $F_k$ is expanded, eventually to all of $J^rHM$,
though possibly with a larger value of $r$, thus showing that $C_k[\h]$
is actually a differential operator of globally bounded order.

\emph{3. Diffeomorphism covariance and the Thomas replacement theorem.} \quad
Now we move on to the \emph{covariance property} of the $C_k$ under
diffeomorphisms, which will be used in several stages. First, fixing the
previously made choice of $y\in M$, we note that the preceding arguments
using the Peetre--Slov\'ak theorem can be repeated for any pair of $y'\in
M$ and $\g_0' = \chi^*\g_0$, where $\chi\colon M\to M$ is some
diffeomorphism such that $\chi(y') = y$, giving rise to differential
orders $r'$ and domains $V_2^{\prime r'} \sse J^{r'} HM$. The
diffeomorphism covariance of $C_k$ then implies that all these
differential orders are the same, $r' = r$, and that the union
$V_3^{\prime r} \sse J^r HM$ of all the $V_2^{\prime r'}$ domains
defines a neighborhood of the $\Diff(M)$-orbit of $j^r(\g_0,0,0) \in J^r
HM$. In fact, $V_3^{\prime r}$ can itself be chosen to be
$\Diff(M)$-invariant (for instance, by taking the union of all
$\Diff(M)$ images of a non-invariant $V_3^{\prime r}$) and a function
$F_k$ satisfying~\eqref{eq:germ-jet-id} uniquely defined on it. The
diffeomorphism covariance of $C_k$ then implies that $F_k$ is itself
$\Diff(M)$-invariant on $V_3^{\prime r}$, in the sense described in
Sect.~\ref{sec:diff}. The case of $\Diff^+(M)$ covariance is handled
in exactly the same way. 

Since diffeomorphisms act transitively on $M$, a diffeomorphism
invariant $V^{\prime r}_3$ would then project down to all of $M$.
Instead, motivated by the desire to keep working in the coordinates
adapted to the local chart $(x^a)$ on $U\sse M$, we choose $V^r_3$
instead to be the intersection of $V_3^{\prime r}$ and the pre-image of
$U$ under the projection $J^rHM\to M$. Then, we have all the needed
hypothesis to apply Proposition~\ref{prp:thomas} to eliminate the
dependence of $F_k$, as a diffeomorphism invariant function, on some of
the coordinates on $V^r_3$. Actually, part of the almost homogeneous
scaling property implies that the functions $H_j$ from
Eq.~\eqref{eq:H-def} are each separately invariant under
diffeomorphisms, so that we can apply Proposition~\ref{prp:thomas} to
each of them individually. Therefore, we can conclude that
\begin{multline}
	g^{-\frac{k(n-2)}{4n}} H_j(x^a, g^{-\frac{1}{n}} g_{ab},
		g^{\frac{1}{n} + \frac{1}{n}|A|} g^{ab,A},
		g^{\frac{1}{n} + \frac{1}{n}|A|} w^A,
		g^{\frac{1}{n}|A|} z^A) \\
	= g^{-\frac{k(n-2)}{4n}} G_j(g^{-\frac{1}{n}} g_{ab},
		g^{\frac{3}{n} + \frac{1}{n}|A|} \bar{S}^{ab(cd,A)},
		g^{\frac{1}{n} + \frac{1}{n}|A|} \bar{w}^A,
		g^{\frac{1}{n}|A|} \bar{z}^A) ,
\end{multline}
where the notation used for the coordinates is explained in
Sect.~\ref{sec:diff} and each $g^{-\frac{k(n-2)}{4n}} G_j$, for
$j=0,\ldots, l-1$, is invariant under the natural action of either
$GL(n)$ (or $GL^+(n)$, depending on which of the cases (a) or (b) we are
dealing with) on its arguments. Notably, $G_j$ depends neither on the
base $(x^a)$ nor on the Christoffel coordinates $(\Gamma^a_{(bc,A)})$.

The invariance properties of $V_3^r$ now tells us that it has the
structure $V^r_3 \cong U \times L_n \times \R^\gamma \times W^r_3$,
where the coordinates $(x^a)$ effect the projection onto the $U$ factor,
the coordinates $(g_{ab})$ or $(g, g^{-\frac{1}{n}}g_{ab})$ effect the
projection onto the $L_n$ factor (the whole space of non-degenerate
bilinear forms on $\R^n$ with Lorentzian signature), the coordinates
$\Gamma^a_{(bc,A)}$ effect the projection on the $\R^\gamma$ argument
and the remaining coordinates $(g^{\bar{3}{n} + \frac{1}{n}|A|}
\bar{S}^{ab(cd,A)}, g^{\frac{1}{n} + \frac{1}{n}|A|} \bar{w}^A,
g^{\frac{1}{n} + \frac{1}{n}|A|} \bar{z}^A)$ effect the projection on
the $W^r_3$ argument, which contains at least the point $(0,0,0)$ and is
invariant under the corresponding action of $GL(n)$ (resp.~$GL^+(n)$).

\emph{4. Invariance under coordinate scaling.} \quad
Next, recall the action of the subgroup of $GL(n)$  (resp.~$GL^+(n)$)
that we called \emph{coordinate scalings} in Sect.~\ref{sec:diff}.
Notice that all the coordinates that the functions $G_j$ depend on have
positive weight with respect to coordinate scalings, with the exception
of $(g^{-\frac{1}{n}} g_{ab}, z)$. For brevity, let us rewrite our
coordinates as $(g, g^{-\frac{1}{n}}g_{ab}, z, q^i)$, with the weight of
the coordinate $q^i$ under coordinate scalings denoted by $d_i > 0$.
Then the invariance of the functions $F_k$ on $V^r_3$ under
diffeomorphisms, and hence coordinate scalings, implies the identity
\begin{align}
\notag
	\mu^{k\frac{(n-2)}{2}}F_k(g, g^{-\frac{1}{n}} g_{ab}, z, q^i)
		&= \mu^{k\frac{(n-2)}{2}} F_k(\mu^{2n} g, g^{-\frac{1}{n}} g_{ab}, z, \mu^{d_i} q^i)\\
\label{eq:scal2-id1}
		&= g^{-\frac{k(n-2)}{4n}} \sum_{j=0}^{l-1}
				\log^j (\mu^{-1} g^{-\frac{1}{2n}})
				G_j(g^{-\frac{1}{n}} g_{ab}, z, \mu^{d_i} q^i)
\end{align}
for any point of $V^r_3$ on its left hand side and any value of $\mu >
0$. As described above, the limit $(g^{-\frac{1}{n}} g_{ab}, z, 0)$ of
the arguments of the functions $G_j$ as $\mu \to 0$ falls within the
domain of the functions $G_j$. Therefore, while the limit of the
left-hand side of~\eqref{eq:scal2-id1} converges to $0$ as $\mu \to 0$,
the right-hand side diverges unless all $G_j = 0$ for $j>0$, so that
$F_k = g^{-\frac{k(n-2)}{4n}} G_0$. The new identity implied by invariance
under coordinate scalings is then
\begin{equation} \label{eq:scal2-id2}
	g^{-\frac{k(n-2)}{4n}}
			G_0(g^{-\frac{1}{n}} g_{ab}, z, q^i)
		= \mu^{-k\frac{(n-2)}{2}} g^{-\frac{k(n-2)}{4n}}
				G_0(g^{-\frac{1}{n}} g_{ab}, z, \mu^{d_i} q^i) .
\end{equation}
Fix some values for the coordinates $(g, g^{-\frac{1}{n}} g_{ab},
z)$ and recall that the point $(g^{-\frac{1}{n}} g_{ab}, z,
0)$ is part of the domain of definition of $G_0$. Since $G_0$ is smooth,
Taylor's theorem allows us to write it as
\begin{equation}
	G_0(g^{-\frac{1}{n}} g_{ab}, z, q^i)
		= \sum_{|I|<N} A_I(g^{-\frac{1}{n}} g_{ab}, z) q^I
			+ O(q^N) ,
\end{equation}
where $I = i_1\cdots i_m$ is a multi-index with respect to the
coordinates $(q^i)$ and $N>0$ is an integer large enough so that
$\langle d, I \rangle = \sum_{j=1}^m d_{i_j} > k$ for any $m = |I| > N$.
Note that the error term $O(q^N)$, for fixed $(q^i)$ mapped to
$(\mu^{d_i} q^i)$ and $\mu \to 0$, is mapped to $O(\mu^{k+1})$ by our
choice of sufficiently large $N$. Thus, using Taylor's theorem, we can
rewrite~\eqref{eq:scal2-id2} as
\begin{multline}\label{eq:scal2-taylor}
	g^{-\frac{k(n-2)}{4n}}
			G_0(g^{-\frac{1}{n}} g_{ab}, z, q^i)
		= \sum_{|I|<N} g^{-\frac{k(n-2)}{4n}}
				A_I(g^{-\frac{1}{n}} g_{ab}, z) q^I
				\mu^{\langle d, I\rangle - k\frac{(n-2)}{2}} \\
			+ \mu^{-\frac{k(n-2)}{2}} O(\mu^{\frac{k(n-2)}{2}+1}) .
\end{multline}
While the left-hand side of~\eqref{eq:scal2-taylor} is bounded as $\mu
\to 0$, the right-hand side diverges unless all $A_I = 0$ for $I$ such
that $\langle d, I\rangle < k(n-2)/2$. If this vanishing condition is
satisfied, the $\mu\to 0$ limits of both sides of~\eqref{eq:scal2-taylor}
exist and give the identity
\begin{equation}\label{eq:ai-poly}
	F_k =
	g^{-\frac{k(n-2)}{4n}}
			G_0(g^{-\frac{1}{n}} g_{ab}, z, q^i)
		= \sum_{\langle d, I \rangle = \frac{k(n-2)}{2}} g^{-\frac{k(n-2)}{4n}}
				A_I(g^{-\frac{1}{n}} g_{ab}, z) q^I .
\end{equation}
At this point, we can once more enlarge the domain of definition of the
function $F_k$, where  the identity~\eqref{eq:germ-jet-id} holds, from
$V^r_3$ to $V^r_4 \sso J^rHM$. The new domain is isomorphic to $V^r_4
\cong U \times L_n \times \R \times W_4 \times \R^\gamma \times
\R^\delta$, where the coordinates $(x^a)$ effect the projection onto the
$U$ factor, the coordinates $(g_{ab})$ or $(g, g^{-\frac{1}{n}} g_{ab})$
effect the projection onto the $L_n$ factor, the coordinate
$(g^{\frac{1}{n}} w)$ effects the projection onto the $\R$ factor, the
coordinate $(z)$ effects the projection onto the $W_4$ factor (which at
least contains the point $(0)$), the coordinates $(\Gamma^a_{(bc,A)})$
effect the projection onto the $\R^\gamma$ factor, and the remaining
coordinates $(g^{\frac{3}{n} + \frac{1}{n}|A|} \bar{S}^{ab(cd,A)},
g^{\frac{1}{n} + \frac{1}{n}|A|} \bar{w}^A, g^{\frac{1}{n} +
\frac{1}{n}|A|} \bar{z}^A)$ effect the projection onto the $\R^\delta$
factor, where the coordinates involving $\bar{w}^A$ and $\bar{z}^A$ with
$|A|=0$ are obviously excluded. Note that $U\times L_n \times \R \times
W_4 \sse HM$ and that $V^r_4$ is simply its pre-image with respect to
the bundle projection $J^rHM \to HM$. The function $F_k$ extends
uniquely from $V^r_3$ to a function on $V^r_4$ that is invariant under
coordinate scalings. The reason we could extend the domain so much,
essentially the factor $W^r_3$ got enlarged to $\R \times W_4 \times
\R^\delta$, is because almost all coordinates, those we labeled by
$(q^i)$ above, had positive degrees with respect to coordinate scalings.
The range of the $(z)$ coordinate is limited to $W_4$ because it is
invariant under coordinate scalings and even under the larger group
$GL(n)$ (resp.~$GL^+(n)$) that acts on the other bundle coordinates.
Also, note that according to Eq.~\eqref{eq:ai-poly} the dependence
of $F_k$ on the $\R \times \R^\gamma \times \R^\delta$ factor in
$V^r_4$, corresponding to the coordinates we labeled by $(q^i)$ above,
is polynomial.

\emph{5. $GL(n)$-equivariance  and polynomial dependence on the metric.} \quad
From the preceding discussion, the function $F_k$, satisfying the
identity~\eqref{eq:germ-jet-id}, is defined on the domain $V^r_4 = U
\times V_4$ and depends only on the coordinates corresponding to the
factor $V_4 = L_n\times W_4 \times \R^\delta$ (where we have grouped all
the $\R\times \R^\gamma \times \R^\delta$ factors together into
$\R^\delta$, implicitly redefining $\delta$).  Moreover, the dependence
on the coordinates on the $\R^\delta$ factor is polynomial, while the
coefficients $g^{-\frac{k(n-2)}{4n}} A_I (g^{-\frac{1}{n}} g_{ab}, z)$ of
these polynomials depend only on the $L_n \times W_4$ factor. It is also
clear from the preceding discussion that each of the factors in $V_4$
carries a tensor density representation of $GL(n)$ (resp.~$GL^+(n)$)
(cf.~Sect.~\ref{sec:iso-tens}), which happens to be trivial on $W_4$.
The space of functions on $V_4$ then itself carries a representation of
$GL(n)$ (resp.~$GL^+(n)$), induced by the pullback of the action on
$V_4$, and the function $F_k$ is invariant under this action. In the
same way, the space $\P^N_\delta$ of polynomials of degree no greater
than $N$ on $\R^\delta$ carries a representation of $GL(n)$
(resp.~$GL^+(n)$),
\begin{equation}
	(uP)(\rho) = P(u^{-1}\rho),
	\quad
	\text{for any $u\in GL(n)$, $P\in \P^N_\delta$ and $\rho \in \R^\delta$,}
\end{equation}
which by elementary reasoning, within the representation theory of
$GL(n)$~\cite{fulton}, is a direct sum of tensor density
representations. Let us group these subrepresentations by tensor rank
and density weight. Therefore, $\P^N_\delta = \bigoplus_j T_j$, where
each $T_j$ is a tensor density representation.

The form that we have reduced $F_k$ to can be described as follows.
Given a point $(\g,\xi,\rho) \in V_4$, the \emph{$A$-coefficients}
$g^{-\frac{k(n-2)}{4n}} A_I(g^{-\frac{1}{n}} g_{ab}, z)$ evaluated at
$(\g,\xi)\in L_n\times W_4$ give a polynomial in $\P^N_\delta$, which is
then evaluated at $\rho\in \R^\delta$. Thus we can think of the
$A$-coefficients as a collection of functions $A_j \colon L_n\times W_4
\to T_j$, with components given by
\begin{equation}\label{eq:A-def}
	(A_j(g_{ab},z))_I = g^{-\frac{k(n-2)}{4n}} A_I(g^{-\frac{1}{n}} g_{ab},z).
\end{equation}
The only way for $F_k$ constructed in this way to be invariant under the
action of $GL(n)$ is for the maps $A_j$ to be equivariant
(cf.~Sect.~\ref{sec:iso-tens}), so that
\begin{multline}
	F_k(u\g,u\xi,u\rho)
		= \sum_j A_j(u\g,u\xi)(u\rho)
		= \sum_j (uA_j(\g,\xi))(u\rho) \\
		= \sum_j A_j(\g,\xi)(u^{-1}u\rho)
		= \sum_j A_j(\g,\xi)(\rho)
		= F_k(\g,\xi,\rho) ,
\end{multline}
for any $u\in GL(n)$ (resp.~$GL^+(n)$) and $(\g,\xi,\rho)\in V_4$.

We are finally in a position to conclude that, for a fixed $\xi \in
W_4$, the map $A_j(-,\xi) \colon L_n \to T_j$ is an \emph{equivariant
tensor density}, in the sense of Definition~\ref{def:eqtd}, and hence
must be of the form dictated by Lemma~\ref{lem:equiv-tens}, which
characterizes all such maps in a way, in view of
Remark~\ref{rem:eqtd-hom}, compatible with our formula~\eqref{eq:A-def}.
In other words, the coefficients of the polynomials $A_j(\g,\xi)$ depend
themselves polynomially on the components $g_{ab}$ and $\eps_{a_1\cdots
a_n}$ of the covariant metric and Levi-Civita tensors, up to an overall
multiple of $g = \left|\det g_{ab}\right|$. If $F_k$ is invariant under
$GL(n)$, then the dependence on $\eps_{a_1\cdots a_n}$ must be trivial,
while it could in general be non-trivial if $F_k$ is invariant only
under $GL^+(n)$. Expanding all the polynomials in $g_{ab}$,
$\eps_{a_1\cdots a_n}$ and $q^i$, all the factors of powers of $g$ must
collectively cancel to preserve invariance of $F_k$ under $GL(n)$
(resp.~$GL^+(n)$). In other words, we can conclude that
\begin{equation}\label{eq:fk-poly}
	F_k = \sum_j a_j(z)
		P_j(g_{ab}, \eps_{a_1\cdots a_n}, \bar{S}^{ab(cd,A)},
			\bar{w}^{A}, \bar{z}^A) , \quad
	\text{with $|A|\ge 1$ in $\bar{z}^A$,}
\end{equation}
where the sum is over a (necessarily finite) basis of polynomials $P_j$,
which consist of linear combinations of tensor contractions of products
of their arguments, with coefficients arbitrarily depending on the $z$
coordinate. In this form, the function $F_k$ is manifestly invariant
under $GL(n)$ (resp.~$GL^+(n)$) transformations.

\emph{6. Global boundedness of differential order.} \quad
To conclude the proof, it remains only to extend the domain $V^r_4$ once
more, this time to all of $J^rHM$, for an appropriate choice of $r$. It
is well known that for a fixed weight $s$ under physical scaling, there
is only a finite number of linearly independent polynomials $P_j$ of
weight $s$ constructed, as described above, from the metric and the
covariant derivatives of the scalar fields $m^2$, $\xi$ and the Riemann
curvature tensor in the form $\bar{S}^{abcd}$, even if the number of the
derivatives $r$ is allowed to be arbitrary%
	\footnote{To see that, consider a monomial of the schematic form \[
	(g_{ab})^{p_g} (\eps_{a_1\cdots a_n})^{p_\eps} \prod_{|A|}
	(\bar{S}^{ab(cd,A)})^{p_{S,|A|}} (\bar{w}^A)^{p_{w,|A|}}
	(\bar{z}^A)^{p_{z,|A|}} , \] necessarily with $p_{z,0}=0$, and note
	that the $p$-exponents must satisfy the constraint
	$\sum_{|A|} [(2+|A|)p_{S,|A|} + (1+|A|)p_{w,|A|} + |A| p_{z,|A|}] =
	s$, due to $s$-homogeneity with respect to physical scalings and
	invariance with respect to coordinate scalings. Since each
	$p$-exponent is non-negative, this implies a bound on the maximum
	value of $|A|$ with a non-zero exponent.}%
~\cite{fkwc}. Let $r_k$ be the
maximum number of derivatives that appear in a basis for these
polynomials $P_j$ when $s=k(n-2)/2$. Then, no matter the original choice of domain $U \sse
M$, the differential operator $C_k$ restricted to it must be of order
$\le r_k$. Thus, we are justified in setting $r = r_k$ in all of the
preceding discussion. The only obstacle that may have prevented us from
extending the domain $V^r_4 \sse J^{r_k}HM$ of the function $F_k$ to all
of the pre-image of $U$ under the projection $J^{r_k}HM \to M$ is the
possibility that $C_k$ would change order on jets whose projections fall
outside $V^{r_k}_4$. However, with the maximal possible order of $C_k$
bounded by $r_k$, this obstacle is now absent. In other words, we can
safely presume that $V^{r_k}_4$ is equal to the pre-image of $U\sse M$
with respect to the projection $J^{r_k}HM \to M$, with $F_k$ retaining
the form~\eqref{eq:fk-poly} on all of its domain. A slightly more
detailed version of this argument would note that the original choice of
the domain $V_1^r$ to be a neighborhood of a point
$j^r_y(\g_0,m^2=0,\xi=0)$ in $J^rHM$ could have equally been chosen to
be a neighborhood of the point $j^r_y(\g_0,m^2=0,\xi_0)$, without
affecting any subsequent arguments. Piecing together $F_k$ over the
extensions of all such neighborhoods gives us a definition of $F_k$ on
the entire pre-image of $U$ under the projection $J^{r_k}HM \to M$ with
the same global order bound $r_k$. Finally, covariance of $C_k$ with
respect to diffeomorphisms requires that the form~\eqref{eq:fk-poly} is
also independent of the domain $U\sse M$.  Thus, we can conclude that
there exists a globally defined smooth bundle map $F_k\colon J^{r_k}HM
\to \R\times M$ over $M$ of the form~\eqref{eq:fk-poly} such that
$C_k[\h](x) = F_k\circ j^{r_k} \h(x)$ for any $x\in M$ and $\h\in
\Secs(HM)$, which concludes the proof.
\end{proof}

\begin{rem}\label{rem:m2xi-param}
We observe, by looking at the \emph{Locality and the Peetre--Slov\'ak theorem}
step of the above proof and also at the proof of
Lemma~\ref{lem:locality}, that one might wonder at the need to
take $m^2$ and $\xi$ as spacetime-dependent fields rather than
constants, as is usually the case. Our arguments still go through, with
only two changes. First, the microlocal hypothesis mentioned in
\ref{def:wick} must be strengthened to require an empty wavefront set
for $\omega(\varphi^k(x))$ as a distribution on $M\times \R^2$ (with the
$\R^2$ factor standing for the parameter space of $m^2$ and $\xi$)
rather than as a distribution on $M$ for any fixed $m^2$ and $\xi$. Note
that the weaker microlocal requirement does not exclude the infinite
family of counterterms of~\cite{tf} that were discussed in the
Introduction, while the stronger one does. Second, we must make use of
the more general version of the Peetre--Slov\'ak theorem for differential
operators with parameters, as in Proposition~\ref{prp:peetre-param} in
Appendix~\ref{app:peetre-param}. To apply that result, we would need to
let $N=M$ and replace the spacetime manifold $M$ by $P = M\times \R^2$,
adding the $(m^2,\xi)$ parameter space. It would then follow from known
information about $C_k$ that it is local with respect to the natural
projection $P\cong \R^2 \times M \to M$, hence satisfying the more
general Peetre--Slov\'ak theorem.
\end{rem}

We end this section with a couple of straight forward but noteworthy
observations. First, it is a direct result of the proof of
Lemma~\ref{lem:locality} that the set of coefficients $\{C_k[\h]\}$ from
Eq.~\eqref{eq:wick-diff} is determined jointly by the entire
families $\{\varphi^k\}$ and $\{\tilde{\varphi}^k\}$ of Wick powers,
rather than depending on each pair $\varphi^k$ and $\tilde{\varphi}^k$
individually. Second, the converse of Theorem~\ref{thm:main} holds as
well. That is, given a family $\{\varphi^k\}$ of locally covariant Wick
powers and a set $\{C_k[\h]\}$ of satisfying the conclusions of
Theorem~\ref{thm:main}, the formula~\eqref{eq:wick-diff} defines another
family $\{\varphi^k\}$ of locally covariant Wick powers.

\section{Discussion}\label{sec:discuss}
In this work, we have characterized admissible finite renormalizations
of Wick powers of a locally covariant quantum scalar field
$\varphi$ on curved spacetimes, with possibly spacetime-dependent mass
$m^2$ and curvature coupling $\xi$. By local covariance, we mean the
axioms of Brunetti, Fredenhagen and Verch~\cite{bfv}. Our work is a
significant technical improvement on the original work of Hollands and
Wald~\cite{hw-wick} on this subject. The main result
(Theorem~\ref{thm:main}) is a slight generalization of that of Hollands
and Wald, yet our hypotheses are significantly more natural and the
proof is greatly simplified and streamlined.

Under standard hypotheses, on Minkowski space, where the curvature
coupling $\xi$ is absent, it is well known that the finite
renormalizations of the Wick powers $\varphi^k$ are restricted to
linear combinations of Wick powers of lower order, with
dimensionful coefficients that are polynomials in $m^2$, with the total
dimension matching that of $\varphi^k$. This is a strong constraint,
because the resulting space of possibilities is finite-dimensional. On
curved spacetimes, as first proven by Hollands and Wald in~\cite{hw-wick},
adding local covariance and some further more technical hypotheses gives
a result of comparable strength. The only modification is that the
coefficients of lower order Wick powers can also depend
polynomially on curvature scalars and analytically on $\xi$, with the
same restriction on their dimensions. The resulting possibilities no
longer form a finite-dimensional space, but a quasi-finite-dimensional
one, in the sense that it is finitely generated under linear
combinations with coefficients analytic in $\xi$. It is worth noting
that the dependence of finite renormalization terms on the background
metric is entirely contained in the curvature scalars, while
their $\xi$-dependent coefficients must be assigned uniformly across all
spacetimes to preserve local covariance.

The hypotheses of Hollands and Wald, briefly recalled in
Definition~\ref{def:wick}, include the requirements of \emph{locality},
\emph{microlocal regularity} and of \emph{continuous} and
\emph{analytic} dependence on the background spacetime metric and
coupling parameters.  Unfortunately, while playing a crucial role in the
existing proof, the analytic dependence hypothesis has been long
considered somewhat unnatural and technically very cumbersome. We have
found that, by using a standard result of differential geometry (the
Peetre--Slov\'ak theorem, cf.~Proposition~\ref{prp:peetre} and
Appendix~\ref{app:peetre-param}), in the presence of the remaining
assumptions, the role of \emph{both} the continuity and analyticity
hypotheses is completely subsumed by that of locality and a strengthened
version of the microlocal regularity condition. We believe the
strengthened, so-called \emph{microlocal spectral condition} is natural
from both physical and geometrical points of view. Physically it
encapsulates the stability of the microlocal properties of Wick powers
under smooth variations of the background geometry. Geometrically, it
provides just the right hypothesis needed to prove the locality of
finite renormalizations of Wick powers, without leaving the realm of
smooth differential geometry. Thus, by replacing the continuity and
analyticity requirements by a more natural hypothesis, our final result
on the characterization of finite renormalizations of Wick powers, as
stated in Theorem~\ref{thm:main}, is essentially identical to that of
Hollands and Wald. The main difference is that arbitrary smooth
dependence on the coupling $\xi$ is now allowed, instead of just
analytic dependence. Another difference is that we have also explicitly
considered weakening covariance to only under orientation preserving
diffeomorphisms, which increases the renormalization freedom to
curvature scalars constructed also with the Levi-Civita tensor and not
just the metric. Finally, we explicitly treat $m^2$ and $\xi$ as
possibly spacetime-dependent parameters, rather than simple constants.
The original proof of Hollands and Wald also treated them as
spacetime-dependent, while restricting to the case of constants in the
statement of their final result. We noted in Remark~\ref{rem:m2xi-param}
how our arguments could be adapted to treating the parameters as
constants throughout.

As was already mentioned, our characterization of finite
renormalizations extends to theories that need only be covariant with
respect to orientation preserving diffeomorphisms. In particular, in
even dimensions, chiral theories (those not invariant under spatial
parity transformations) could be admissible. While our result does not
contain any surprises, it is important to have a rigorous statement on
the complete range of possibilities. In particular, suppose that a
classical parity invariant theory is perturbatively quantized using a
chiral renormalization scheme. The knowledge of a complete
classification of finite renormalizations is then required to decide
whether there exists a different renormalization scheme that gives a
parity invariant quantization.

Another advantage of our proof is the clear separation between the
applications of the locality, microlocal regularity, covariance and
scaling hypotheses. We make a particular distinction between
\emph{physical scalings} (those resulting from a rescaling of the
metric) and \emph{coordinate scalings} (those resulting from the local
action of some diffeomorphisms). We believe that structuring the proof
in this way makes it significantly easier to generalize the result to
other types of tensor or spinor fields, a task that is yet to be
seriously taken up in the literature on locally covariant quantum field
theory, which is in significant part likely due to the complexity of and
the unnatural hypotheses needed in the original proof of Hollands and
Wald. In particular, it is likely that the crucial step in limiting the
finite renormalization freedom to a quasi-finite-dimensional space is to
carefully balance the covariance and scaling properties, such that there
exists a coordinate system on the jets of background fields, like the
\emph{rescaled curvature coordinates} that we identified in
Sect.~\ref{sec:coords}, where all coordinates corresponding higher
derivatives have positive weight under a combination of the physical and
coordinate scalings.

Another direction in which our main result could be generalized is to
consider Wick powers that included derivatives of fields. Our proof
should extend without problems. The main difference would be that the
finite renormalization coefficients $C_k$ could then be tensor- instead
of scalar-valued, since Wick powers with derivatives could
themselves be tensor fields. This difference would affect the part of
our proof where we make use of $GL(n)$-equivariance to fix the form of
the $A_j$ coefficients, which could be mixed densitized tensors.
Fortunately, the main technical result on the classification of
equivariant tensor densities, as stated in
Lemma~\ref{lem:equiv-tens}, is sufficiently general to apply to
that case as well, since introducing densitization erases the
distinction between covariant and contravariant indices.

Let us also say something about time-ordered products. Hollands and
Wald also gave a sketch of the proof of the characterization of finite
renormalizations of time-ordered products~\cite[Thm.~5.2]{hw-wick}, under
the same hypothesis as their result about Wick powers. As they
point out, the main difference with the case of Wick powers is in the
structure of coefficients that are analogous to the $C_k$, which become
distributions on multiple copies of the spacetime manifold. The
arguments, which we encapsulated in Lemma~\ref{lem:locality}, applying
microlocal arguments to restrict the wavefront set of these
distributions would have to be generalized accordingly. After that
point, the proof of Theorem~\ref{thm:main}, would apply without
essential modifications. Thus, our methods should generalize to
time-ordered products as well.

Finally, we note that, although we believe that our \emph{parametrized
microlocal spectral condition} (Definition~\ref{def:wick}(iv)) does hold
for standard locally covariant Hadamard parametrix prescription for
defining Wick powers, we have not given a proof. In fact a complete
proof of the validity of the continuous and analytic dependence for the
Hadamard parametrix prescription was not given in the original work of
Hollands and Wald either~\cite[Sec.~5.2]{hw-wick} and only appeared in
the later work~\cite{hw2}. Closing this gap with a complete and precise
proof is a worthwhile goal for future work. In fact, to be of greatest
use for the characterization of finite renormalizations of time-ordered
products and of Wick products with derivatives, we would need a proof of
validity of a parametrized microlocal spectral condition for multi-local
fields as well. Such a condition might be reasonably stated as follows,
echoing~\cite[Eq.~(46)]{hw-wick} which considered a similar question for
the analytic wavefront set. Let $\Phi(x_1,\ldots,x_k)$ be a locally
covariant multi-local field that already satisfies the standard
(unparametrized) microlocal spectral condition. Then, let $(M,\h)$ be a
background geometry and let
\begin{equation}
	\Gamma^\Phi(M^k,\h)
		= \overline{\bigcup_{\omega} \mathrm{WF}(\omega(\Phi_{(M,\h)}(-))} 
			\setminus \{0\}
		\sse (T^*M^k)\setminus \{0\},
\end{equation}
where the union is taken over all Hadamard states $\omega$ on
$\W(M,\h)$. Consider also a smooth $m$-parameter compactly supported
variation $\H(s,x)$ of $\h(x)$, together with the accompanying algebra
isomorphisms $\tau^s_\mathrm{ret}\colon \W(M,\H_s) \to \W(M,\h)$. If
$\omega$ is any Hadamard state on $\W(M,\h)$, then we would like to
require that the wavefront set of $E^\Phi_\omega(s,x_1,\ldots,x_k) =
\omega(\tau^s_\mathrm{ret} \circ \Phi_{(M,\H_s)}(x_1,\ldots,x_k))$ as a
distribution on $\R^m \times M^k$ satisfies
\begin{multline}
	\mathrm{WF}(E^\Phi_\omega) \sse
		\{ (s,\sigma, x_1,p_1, \ldots, x_n,p_n) \in T^*(\R^m\times M^k) \\
			\mid (x_1,p_1, \ldots, x_k, p_k) \in \Gamma^\Phi(M,\H_s) \} .
\end{multline}
That is, $E^\varphi_\omega$ or any of its derivatives can be restricted
to the submanifold of $\R^m\times M^k$ given by a fixed value of $s$,
and that restriction has precisely the wavefront set expected of a
locally covariant field on $(M,\H_s)$ satisfying the standard microlocal
spectrum condition.

We now give a sketch of an argument for establishing the validity of our
parametrized microlocal spectrum condition in the simplest case of the
Wick monomial $\phi^2_H(x)$, as constructed using the Hadamard
parametrix regularization method by Hollands and Wald. (The authors are
grateful to an anonymous referee for suggesting it.) In this elementary
situation, the only thing to prove is the joint smoothness of
$\omega_s(\phi_{H_s}^2(x))$ as a function of the position $x$ and the
parameters $s$, with $H_s$ the local Hadamard parametrices of the
compactly supported background geometry variation $\H_s$ and the
Hadamard states $\omega_s = \omega \circ \tau^s_\mathrm{ret}$ are
defined as in the preceding paragraph. As a matter of fact, one could
more generally prove that  $f(s,x,y) := \omega_s (\phi(x)\phi(y)) -
H_s(x,y)$ is a jointly smooth function of $s,x,y$, since
$\omega_s(\phi_{H_s}^2(x))$ is obtained just by taking $x=y$ in the
difference above. The 2-point function $\omega_s(\phi(x)\phi(y))$ is a
global distributional bisolution, while the symmetric distribution
$H_s(x,y)$ is locally defined through a well-known procedure and, as a
parametrix, only satisfies the Klein--Gordon equation (in either $x$ or
$y$) up to a smooth error term, which we would need to show is also
jointly smooth in $s$. So, the function $f(s,x,y)$ will satisfy a pair
of Klein--Gordon equations,
\begin{align}
\label{eq:KG1}
	(\square_{\H_s} - m^2 - \xi R_{\H_s})_x f(s,x,y) &= g(s,x,y) \\
\label{eq:KG2}
	\text{and} \quad
	(\square_{\H_s} - m^2 - \xi R_{\H_s})_y f(s,x,y) &= g(s,y,x) ,
\end{align}
with some jointly smooth $g(s,x,y)$ defined on the same neighborhood of
the diagonal $x=y$ as $H_s(x,y)$. The argument would conclude by
determining a precise form of $f(s,x,y)$ near a Cauchy surface in the
past of the compact region $O \subset M$, where the variation $\H_s$
differs from the reference geometry $\H_0$, and showing the existence of
a unique solution of the above equations, which is moreover jointly
smooth in $s$, $x$ and $y$ and necessarily coincides with $f(s,x,y)$.

We should note that, though the main ideas are clear, the above argument
features some technical difficulties that it would take a separate paper
to fully explore. For instance, it is not immediately clear which result
from PDE theory would assure the existence, uniqueness and smooth
parameter dependence of the solutions of Eqs.~\eqref{eq:KG1}
and~\eqref{eq:KG2}, all rather delicate questions, especially in the
context of partial rather than ordinary differential equations. In fact
even establishing the joint smoothness of $g(s,x,y)$ goes beyond
elementary facts about Hadamard parametrices (though some relevant
arguments were already provided in~\cite[Prop.~4.1]{hw2}
and~\cite[App.~A]{zahn}). Also, even if the preceding issues are
resolved by a clever use of standards results for hyperbolic PDEs, the
domain on which Eqs.~\eqref{eq:KG1} and~\eqref{eq:KG2} are defined may
not be globally hyperbolic in any meaningful way, because $g(s,x,y)$
would be defined only on a neighborhood of the diagonal in $M\times M$.

Thus, we leave the investigation of the above generalized parametrized
microlocal spectral condition, of more general types of fields, of Wick
powers with derivatives and of time-ordered products for future work.

\section*{Acknowledgments}
The authors thank Nicola Pinamonti and Romeo Brunetti for their comments
on an earlier version of the manuscript. They also thank an anonymous
referee for pointing out a serious omission from our statement of the
Peetre--Slov\'ak theorem in a previous version of the manuscript.

\appendix

\section{Peetre--Slov\'ak's theorem with parameters}\label{app:peetre-param}
Below, we first make some remarks about how the \emph{weak regularity}
hypothesis (Definition~\ref{def:reg}) in Proposition~\ref{prp:peetre}
can be justified, despite the stronger \emph{regularity} hypothesis that
is usually required~\cite{slovak,kms,ns}. Then, we state a more general
version of Proposition~\ref{prp:peetre}, in which the notion of locality
is generalized to accommodate parameters. The usually complicated way in
which this more general locality condition is stated is clarified
through examples.

The paper~\cite{ns} gives an excellent, self-contained and
straight-forward proof of the version of the Peetre--Slov\'ak
theorem~\cite[Thm.~3.1]{ns} that we state in
Proposition~\ref{prp:peetre}, with the exception that it requires the
stronger \emph{regularity} instead of the \emph{weak regularity}
hypothesis. There, the regularity hypothesis is used in exactly two
places: (a) It is used once directly in the proof of Theorem~3.1, to
show that the map $D[\phi] = d\circ j^k\phi$ factors through a
\emph{smooth} map $d$ on the space of $k$-jets. In that instance, the
smooth ``universal family of $k$-jets'', which establishes the
smoothness of $d$ all at once, can easily be replaced by a related smooth
compactly supported variation that establishes the smoothness of $d$ in
a compact neighborhood of any point of its domain. However, since that
can be done for any point in the domain of $d$, the desired global
smoothness of $d$ is immediate. (b) Regularity is used once more in the
proof of the intermediate Lemma~2.4, to show that $D[\phi] = d\circ
j^k\phi$ factors through some \emph{finite order} jet space. There,
regularity is called upon when $D$ acts on smooth families of sections
$f$ and $h$ that have been constructed to have controlled behavior on a
compact set $K$ by appealing to the \emph{Whitney extension
theorem}~\cite[Sec.~1.2]{ns}, where $K = \{z_k\}_{k=0}^\oo$ consists of
the points of a convergent sequence. Note that proof of that Lemma~2.4
goes through even with only weak regularity, provided that $f$ and $h$
could be constructed as smooth compactly supported variations, rather
than smooth families. The Whitney extension theorem, as stated
in~\cite[Sec.~1.2]{ns}, which constructs smooth extensions of
consistently specified jet data on an arbitrary compact set, can only
produce families of sections rather than compactly supported variations.
However, it is well known that Whitney's extension theorem can be
strengthened~\cite{bierstone} to construct smooth extensions of
consistently specified jet data on arbitrary closed rather than just
compact sets. Using this strengthened version, it is easily seen that
the above mentioned $f$ and $h$ can be constructed as smooth compactly
supported variations with specified behavior both on the compact set $K$
and on the complement of any open neighborhood of $K$ with compact
closure, thus showing that the proof of~\cite[Lem.~2.4]{ns} can be
completed with only weak regularity.

The Peetre--Slov\'ak theorem stated in Proposition~\ref{prp:peetre} may
be made significantly stronger by generalizing the admissible notion of
locality. Let us now introduce the language needed to state the stronger
version in a precise form. In the following, $\sigma\colon E\to N$ and
$\rho\colon F\to M$, are two smooth bundles, where we have explicitly
written the canonical projections, and we consider a map $D\colon
\Secs(E\to N) \to \Secs(F\to M)$ between smooth sections of these
bundles. We intend here to give a precise mathematical meaning to the
statement that \emph{$D$ is local}. Before defining the most general
version of locality (cf.~\cite[\textsection~18.16]{kms}), we consider
several motivating cases of increasing complexity.

\emph{Case  $N = M$.} We say that $D$ is \emph{local} when the value
$\phi(x)$, for $\phi = D[\psi] \in \Secs(F\to M)$, depends only on the
germ of $\psi\in \Secs(E\to M)$ at $x \in M$. This version of locality
is already sufficient for Propositions~\ref{prp:peetre0}
and~\ref{prp:peetre}. We can loosen this notion of locality in several
ways.

\emph{Case  $N \neq M$.} We may agree that $\phi(x)$, for $\phi =
D[\psi] \in \Secs(F\to M)$ and $x\in M$, may depend only on the germ of
$\psi\in \Secs(E\to N)$ at $y\in N$, with some fixed relationship $y =
\chi(x)$, where $\chi\colon M \to N$ is some diffeomorphism. We then say
that $D$ is \emph{$\chi$-local}.

\emph{Case  $N \neq M$ and $D$ depends on external parameters.} We can
introduce a bundle $\pi\colon P \to M$, where the manifold $P$ is
interpreted as ``$M$ with parameters.'' Then, allowing $D$ to depend on
parameters means that $D$ really maps sections of $E\to N$ to sections
of the pullback bundle $\pi^* F \to P$, interpreted as ``$F$ with
parameters.'' Let us briefly recall that, given a bundle $F\to M$ and a
map $\pi\colon P\to M$, the pullback bundle $\pi^*F \to P$ is uniquely
defined by the existence of a bundle morphism $\tilde{\pi}\colon \pi^*F
\to F$ that is a fiber-wise isomorphism and that makes the following
diagram commute
\begin{equation}\label{eq:pbb-def}
\begin{tikzcd}
	\pi^* F \ar{d} \ar{r}{\tilde{\pi}} & F \ar{d}{\pi} \\
	P \ar{r}{\pi} & M
\end{tikzcd} ~~ .
\end{equation}
Pre-composing a section of $\pi^*F \to P$ with a section of $P\to M$
then yields a section of $F\to M$ given by a particular choice of
parameters. Denoting $\eta = \chi\circ \pi$, we call the map $D\colon
\Secs(E\to N) \to \Secs(\pi^*F \to P)$ \emph{$\eta$-local} when
$\phi(x,p) = D[\psi](x,p)$, with $(x,p)\in P$ and $\pi(x,p) = x \in M$,
depends only on the germ of $\psi$ at $y = \eta(x,p) = \chi(x) \in N$.
Note that the total space of the bundle ``$F$ with parameters'' can be
expressed as the fibered product $\pi^* F \cong F {}_\rho{\times}_\pi P$
over $M$ (where we have explicitly named the $\rho\colon F\to M$ bundle
projection), which completes the pullback diagram
\begin{equation}\label{eq:fib-prod}
\begin{tikzcd}
	F {}_\rho{\times}_\pi P \ar{d} \ar{r} & P \ar{d}{\pi} \\
	F \ar{r}{\rho} & M
\end{tikzcd} ~~ .
\end{equation}
We can illustrate all of the above maps in the diagram
\begin{equation}\label{eq:strict-params}
\begin{tikzcd}
	E \ar{d} &
		F {}_\rho{\times}_\pi P \ar{d} \ar{r} &
		F \ar{d}[swap]{\rho} \\
	N \ar[bend left,dotted]{u}[pos=.45]{\psi} &
		P \ar{r}[swap]{\pi}
			\ar{l}{\eta}
			\ar[bend left,dotted]{u}{\phi=D[\psi]} &
		M \ar[bend left,swap]{ll}{\chi}
			\ar[bend right,dotted,swap]{l}{\tau}
			\ar[bend right,dotted,swap]{u}[pos=.45]{\phi\circ\tau}
\end{tikzcd} ~~ ,
\end{equation}
where all the solid arrows commute, the dotted arrows denote bundle
sections, with $\tau\colon M \to P$ denoting a particular ``choice of
parameters,'' and $\phi\circ \tau$ was silently composed with the
projection $F {}_\rho{\times}_\pi P \to F$.

\emph{General case.}
Finally, it is possible to relax the requirement that the map
$\eta\colon P \to N$ factors as illustrated in
diagram~\eqref{eq:strict-params}. The dimension of $P$ could exceed that
of $N$ and $\eta$ need not be a surjection, not even a submersion.
Omitting the structure of the right square of
diagram~\eqref{eq:strict-params}, we also replace $F {}_\rho{\times}_\pi
P$ by a simple bundle $F\to P$, without requiring it to have the
structure of a fibered product. So, given bundles $E\to N$ and $F\to P$,
together with a smooth map $\eta \colon P\to N$, a map $D\colon
\Secs(E\to N) \to \Secs(F\to P)$ is called \emph{$\eta$-local} if
$\phi(x) = D[\psi](x)$, $x\in P$, depends only on the germ of $\psi$ at
$y = \eta(x) \in N$. We can illustrate this situation by the diagram
\begin{equation}\label{eq:params}
\begin{tikzcd}
	E \ar{d} &
		{~~~ F ~~~} \ar{d} \\
	N \ar[bend left,dotted]{u}[pos=.45]{\psi} &
		P \ar{l}{\eta}
			\ar[bend left,dotted]{u}[pos=.45]{\phi=D[\psi]}
\end{tikzcd} ,
\end{equation}
which should be thought of as exactly analogous to
diagram~\eqref{eq:strict-params}, but with the right square missing.
This the rather weak notion of $\eta$-locality, with a small additional
hypothesis ($\eta$ non-locally constant), together with the condition of
weak regularity (Definition~\ref{def:reg}) is actually sufficient for
the more general version of Peetre--Slov\'ak's theorem.

\begin{prop}[Peetre--Slov\'ak's Theorem~{\cite[\textsection~19.10]{kms}}]\label{prp:peetre-param}
Let $F\to P$, $E\to N$ be smooth bundles and $\eta\colon P\to N$ a
non-locally constant\/%
	\footnote{By \emph{non-locally constant} we mean that for every open
	$U\sse P$ the image $\eta(U)$ contains at least two points.} %
smooth map, with the interpretation as in diagram~\eqref{eq:params},
and $D\colon \Secs(E\to N) \to \Secs(F\to P)$ be an $\eta$-local and
weakly regular map. Then,
for every compact $K\subseteq P$ and $\psi\in \Secs(E\to N)$, there
exists an integer $r$, an open neighborhood $U\sse J^r(E\to N)$ of
$j^r\psi(N)\sso U$, with $U_K\sse U$ the subset projecting onto
$\eta(K)$, and a function $d\colon U_K \to F$ that commutes with all the
projections, as illustrated by the diagram
\begin{equation}
\begin{tikzcd}
	\llap{$J^r(E\to N) \supseteq {}$} U_K \ar{d} \ar{r}{d}
		& F\rlap{$|_K$} \ar{d} \\
	\llap{$N \supseteq {}$} \eta(K)
		& \ar{l}{\eta} K \rlap{${} \sse P$}
\end{tikzcd} \rlap{\quad\qquad ,}
\end{equation}
such that $D[\xi](x) = d\circ j^r\xi(x)$ for any $\xi\in \Secs(E\to N)$
with $j^r\xi(N) \sso U$. In other words, $D$ is a differential operator
of locally finite order, where locality is with respect to compact
subsets of $P$ and compact open neighborhoods in $\Secs(E\to N)$.
\end{prop}
\begin{proof}[Sketch of proof]
With the definitions as discussed above, the proposition is essentially
a restatement of Theorem~19.10 of~\cite{kms}, which follows directly
from Theorem~19.7 and Corollary~19.8 that precede it. We refer the
reader to the book~\cite{kms} for full details. Let us simply mention
that, in general outline, the proof proceeds by contradiction. If $D$
depended non-trivially on an infinite number of derivatives of its
argument, then it would be possible to engineer a smooth section $\psi$
such that $D[\psi]$ could not itself be smooth. While the proofs
in~\cite{kms} rely on regularity instead of weak regularity, the weaker
hypothesis is actually sufficient, as we have discussed earlier.
\end{proof}

\bibliographystyle{utphys-alpha}
\bibliography{paper-peetre}

\providecommand{\href}[2]{#2}\begingroup\raggedright\begin{thebibliography}{10}

\bibitem{at-spinors}
I.~M. Anderson and C.~G. Torre, ``Two component spinors and natural coordinates
  for the prolonged {Einstein} equation manifolds,'' tech. rep., Utah State
  University, 1994.

\bibitem{ta-sym}
I.~M. Anderson and C.~G. Torre, ``Classification of local generalized
  symmetries for the vacuum {Einstein} equations,''
  \href{http://dx.doi.org/10.1007/bf02099248}{{\em Communications in
  Mathematical Physics} {\bfseries 176} (1996) 479--539},
  \href{http://arxiv.org/abs/gr-qc/9404030}{{\ttfamily arXiv:gr-qc/9404030}}.

\bibitem{ado}
P.~G. Appleby, B.~R. Duffy, and R.~W. Ogden, ``On the classification of
  isotropic tensors,'' \href{http://dx.doi.org/10.1017/s0017089500006832}{{\em
  Glasgow Mathematical Journal} {\bfseries 29} (1987) 185--196}.

\bibitem{bierstone}
E.~Bierstone, ``Differentiable functions,''
  \href{http://dx.doi.org/10.1007/bf02584636}{{\em Boletim da Sociedade
  Brasileira de Matem\'{a}tica} {\bfseries 11} (1980) 139--189}.

\bibitem{bfk}
R.~Brunetti, K.~Fredenhagen, and M.~K\"{o}hler, ``The microlocal spectrum
  condition and {Wick} polynomials of free fields on curved spacetimes,''
  \href{http://dx.doi.org/10.1007/bf02099626}{{\em Communications in
  Mathematical Physics} {\bfseries 180} (1996) 633--652}.

\bibitem{bf}
R.~Brunetti and K.~Fredenhagen, ``Microlocal analysis and interacting quantum
  field theories: Renormalization on physical backgrounds,''
  \href{http://dx.doi.org/10.1007/s002200050004}{{\em Communications in
  Mathematical Physics} {\bfseries 208} (2000) 623--661}.

\bibitem{bfv}
R.~Brunetti, K.~Fredenhagen, and R.~Verch, ``The generally covariant locality
  principle -- {A} new paradigm for local quantum field theory,''
  \href{http://dx.doi.org/10.1007/s00220-003-0815-7}{{\em Communications in
  Mathematical Physics} {\bfseries 237} (2003) 31--68}.

\bibitem{fkwc}
S.~A. Fulling, R.~C. King, B.~G. Wybourne, and C.~J. Cummins, ``Normal forms
  for tensor polynomials. {I.} {The} {Riemann} tensor,''
  \href{http://dx.doi.org/10.1088/0264-9381/9/5/003}{{\em Classical and Quantum
  Gravity} {\bfseries 9} (1992) 1151--1197}.

\bibitem{fulton}
W.~Fulton, \href{http://dx.doi.org/10.1017/CBO9780511626241}{{\em Young
  Tableaux: With Applications to Representation Theory and Geometry}}, vol.~35
  of {\em London Mathematical Society Student Texts}.
\newblock Cambridge University Press, Cambridge, 1996.

\bibitem{gs1}
I.~M. {Gel'fand} and Z.~Y. Shapiro, ``Homogeneous functions and their
  extensions,'' {\em Uspekhi Matematicheskikh Nauk} {\bfseries 10} (1955)
  3--70. \url{http://mi.mathnet.ru/eng/umn7998}.

\bibitem{gs2}
I.~M. {Gel'fand} and G.~E. Shilov, {\em Generalized functions. {V}ol. {I}:
  {P}roperties and operations}.
\newblock Academic Press, New York, 1964.

\bibitem{gw}
R.~Goodman and N.~R. Wallach,
  \href{http://dx.doi.org/10.1007/978-0-387-79852-3}{{\em Symmetry,
  Representations, and Invariants}}, vol.~255 of {\em Graduate Texts in
  Mathematics}.
\newblock Springer, New York, 2009.

\bibitem{hw-wick}
S.~Hollands and R.~M. Wald, ``Local {Wick} polynomials and time ordered
  products of quantum fields in curved spacetime,''
  \href{http://dx.doi.org/10.1007/s002200100540}{{\em Communications in
  Mathematical Physics} {\bfseries 223} (2001) 289--326},
  \href{http://arxiv.org/abs/gr-qc/0103074}{{\ttfamily arXiv:gr-qc/0103074}}.

\bibitem{hw2}
S.~Hollands and R.~M. Wald, ``Existence of local covariant time ordered
  products of quantum fields in curved spacetime,''
  \href{http://dx.doi.org/10.1007/s00220-002-0719-y}{{\em Communications in
  Mathematical Physics} {\bfseries 231} (2002) 309--345},
  \href{http://arxiv.org/abs/gr-qc/0111108}{{\ttfamily arXiv:gr-qc/0111108}}.

\bibitem{Iyer-Wald94}
V.~Iyer and R.~M. Wald, ``Some properties of the noether charge and a proposal
  for dynamical black hole entropy,''
  \href{http://dx.doi.org/10.1103/physrevd.50.846}{{\em Physical Review D}
  {\bfseries 50} (1994) 846--864},
  \href{http://arxiv.org/abs/gr-qc/9403028}{{\ttfamily arXiv:gr-qc/9403028}}.

\bibitem{kms}
I.~Kola{\v{r}}, P.~W. Michor, and J.~Slov{\'{a}}k, {\em Natural Operations in
  Differential Geometry}.
\newblock Electronic Library of Mathematics. Springer, 1993.

\bibitem{ns}
J.~Navarro and J.~B. Sancho, ``{Peetre-Slov\'ak's} theorem revisited,'' 2014.
\newblock \url{http://arxiv.org/abs/1411.7499}.

\bibitem{olver}
P.~J. Olver, \href{http://dx.doi.org/10.1007/978-1-4684-0274-2}{{\em
  Applications of Lie groups to differential equations}}, vol.~107 of {\em
  Graduate Texts in Mathematics}.
\newblock Springer-Verlag, New York, second~ed., 1993.

\bibitem{Peetre1959}
J.~Peetre, ``Une caract\'{e}risation abstraite des op\'{e}rateurs
  diff\'{e}rentiels,'' {\em Mathematica Scandinavica} {\bfseries 7} (1959)
  211--218. \url{http://eudml.org/doc/165715}.

\bibitem{Peetre1960}
J.~Peetre, ``R\'{e}ctification \`{a} l'article : ``{Une} caract\'{e}risation
  abstraite des op\'{e}rateurs diff\'{e}rentiels'','' {\em Mathematica
  Scandinavica} {\bfseries 8} (1960) 116--120.
  \url{http://eudml.org/doc/251805}.

\bibitem{penrose-spinors}
R.~Penrose, ``A spinor approach to general relativity,''
  \href{http://dx.doi.org/10.1016/0003-4916(60)90021-x}{{\em Annals of Physics}
  {\bfseries 10} (1960) 171--201}.

\bibitem{shelkovich}
V.~M. Shelkovich, ``Associated and quasi associated homogeneous distributions
  (generalized functions),''
  \href{http://dx.doi.org/10.1016/j.jmaa.2007.04.069}{{\em Journal of
  Mathematical Analysis and Applications} {\bfseries 338} (2008) 48--70}.

\bibitem{slovak}
J.~Slov\'{a}k, ``Peetre theorem for nonlinear operators,''
  \href{http://dx.doi.org/10.1007/bf00054575}{{\em Annals of Global Analysis
  and Geometry} {\bfseries 6} (1988) 273--283}.

\bibitem{thomas}
T.~Y. Thomas, {\em Differential invariants of generalized spaces}.
\newblock CUP, Cambridge, 1934.

\bibitem{tf}
W.~Tichy and E.~Flanagan, ``How unique is the expected stress-energy tensor of
  a massive scalar field?,''
  \href{http://dx.doi.org/10.1103/physrevd.58.124007}{{\em Physical Review D}
  {\bfseries 58} (1998) 124007},
  \href{http://arxiv.org/abs/gr-qc/9807015}{{\ttfamily arXiv:gr-qc/9807015}}.

\bibitem{torre-spinors}
C.~G. Torre, ``Spinors, jets, and the {Einstein} equations,'' in {\em The Sixth
  Canadian Conference on General Relativity and Relativistic Astrophysics},
  S.~P. Braham, J.~D. Gegenberg, and R.~J. McKellar, eds., vol.~15 of {\em
  Fields Institute Communications}, pp.~125--136.
\newblock AMS, Providence, RI, 1997.
\newblock \href{http://arxiv.org/abs/gr-qc/9508005}{{\ttfamily
  arXiv:gr-qc/9508005}}.

\bibitem{weyl}
H.~Weyl, {\em The Classical Groups: Their Invariants and Representations}.
\newblock Princeton University Press, 1997.

\bibitem{zahn}
J.~Zahn, ``Locally covariant charged fields and background independence,''
  \href{http://dx.doi.org/10.1142/s0129055x15500178}{{\em Reviews in
  Mathematical Physics} {\bfseries 27} (2015) 1550017},
  \href{http://arxiv.org/abs/1311.7661}{{\ttfamily arXiv:1311.7661}}.

\end{thebibliography}\endgroup

\end{document}